\documentclass[twocolumn]{IEEEtran}

\usepackage{amssymb}
\usepackage{amsmath}
\usepackage{booktabs}
\usepackage{amsfonts}
\usepackage{amsthm}
\usepackage{xfrac}
\usepackage{color}
\usepackage{array}
\usepackage{graphicx}
\usepackage[font=footnotesize]{caption}
\usepackage[font=footnotesize]{subcaption}
\usepackage{cite}
\usepackage{wasysym}
\usepackage{dsfont}
\usepackage{cases}
\usepackage{siunitx}
\usepackage[makeroom]{cancel}
\newtheorem{theorem}{Theorem}

\newtheorem{lemma}[theorem]{Lemma}
\DeclareMathOperator*{\ld}{ld}

\begin{document}

\title{How Reliable and Capable is Multi-Connectivity?}
\author{Albrecht Wolf, Philipp Schulz, Meik D{\"o}rpinghaus,~\IEEEmembership{Member,~IEEE},\\Jos\'{e} C\^{a}ndido Silveira Santos Filho,~\IEEEmembership{Member,~IEEE}, and Gerhard Fettweis,~\IEEEmembership{Fellow,~IEEE}
	\thanks{This work has been submitted to the IEEE for possible publication. Copyright may be transferred without notice, after which this version may no longer be accessible. The material in this work has been presented in part at the IEEE GLOBECOM, Singapore, Dec. 2017 \cite{Wolf2017_G1}.}
	\thanks{This work was supported by the Federal Ministry of Education and Research within the programme ``Twenty20 - Partnership for Innovation" under contract 03ZZ0505B - ``fast wireless".}	
	\thanks{A.~Wolf, P.~Schulz, M.~D{\"o}rpinghaus, and G.~Fettweis are with the Vodafone Chair Mobile Communications Systems, Technische Universit\"{a}t Dresden, Dresden, Germany, E-mails: \{albrecht.wolf, philipp.schulz2, meik.doerpinghaus, gerhard.fettweis\}@tu-dresden.de.}
\thanks{J.~C.~S.~Santos Filho is with the Department of Communications, School of Electrical and Computer Engineering, University of Campinas (UNICAMP), Campinas--SP, Brazil, E-mails: candido@decom.fee.unicamp.br.}}

\maketitle
%
%
\begin{abstract}
Multi-connectivity (MCo) is considered to be a key strategy for enabling reliable transmissions and enhanced data rates in fifth-generation mobile networks, as it provides multiple links from source to destination. In this work, we quantify the communication performance of MCo in terms of outage probability and throughput. For doing so, we establish a simple, yet accurate analytical framework at high signal-to-noise ratio (SNR), in which the number of links, the spectral efficiency, the path loss, and the SNR are incorporated, giving new insights into the potentials of MCo as compared with single-connectivity (SCo). These are our main contributions: (1) finding the exact coding gain of the outage probability for parallel block-fading channels; (2) quantifying the performance improvement of MCo over SCo in terms of SNR gain; and (3) comparing optimal and suboptimal combining algorithms for MCo at the receiver side, namely joint decoding, selection combining, and maximal-ratio combining, also in terms of SNR gain. Additionally, we apply our analytical framework to real field channel measurements and thereby illustrate the potential of MCo  to achieve high reliability and high data rates in real cellular networks.
\end{abstract}
%
\section{Introduction}
\label{Introduction}
Fifth-generation mobile networks (5G) will face several challenges to cope with emerging application scenarios~\cite{Schulz2017} in the context of ultra-reliable low latency communications (URLLC) such as mission critical industrial automation or communications for vehicular coordination, which require an extremely high reliability (e.g., frame error rates of $10^{-9}$ or $10^{-5}$, respectively) while simultaneously providing low latency (e.g., end-to-end delay of 1~ms). These requirements pose a massive challenge on the physical layer. In fourth-generation mobile networks (4G), reliability is obtained by the hybrid automatic repeat request procedure, which retransmits erroneously received packets. However, the tight timing constraint of URLLC does not endorse multiple retransmissions. \\ \indent  
Multi-connectivity (MCo) is a promising tool for boosting the reliability and capacity of wireless networks~\cite{Oehemann2017}. Firstly, it provides a flexible communication framework that can trade diversity for multiplexing via multiple routes to the destination. Secondly, MCo architectures can use different carrier frequencies, such that multiple copies of the same information can, in the best case, be delivered within a single time slot. However intuitive such a diversity-multiplexing tradeoff may be, some fundamental questions remain open regarding the potentials of MCo as compared with single-connectivity (SCo):
\begin{enumerate}
\item Given a target (fixed) spectral efficiency, how much transmit power can be saved while achieving a same outage probability at high SNR?
	\item Given a target (fixed) outage probability, how much transmit power can be saved while achieving a same throughput at high SNR?
	\item How those savings vary with the level of the target metric and with the number of connections and topology?
\end{enumerate}
In this work, we answer the above questions by considering both optimal and suboptimal combining algorithms at the receiver side. The former is provided by joint decoding (JD). For the latter, we consider standard diversity-combining methods, namely maximal-ratio combining (MRC) or selection combining (SC). To answer the referred questions, we derive exact integral-form expressions for the outage probability and throughput of each investigated system setup. More importantly, we obtain corresponding asymptotic, closed-form expressions at high signal-to-noise-ratio (SNR) that shed light on the problem, thereby providing the answers we look for.

Next, we briefly describe MCo, the combining algorithms, and then outline our approach and our contributions.
\subsection{Multi-Connectivity}
The MCo concept refers to any system architecture where users are simultaneously connected via multiple communication links. With respect to URLLC, it is most desirable to transmit the same data redundantly (diversity) over independent fading channels in a single time slot. Microdiversity~\cite{Molisch2012}, including spatial and frequency diversity, is well suited to combat small-scale fading, while satisfying the tight latency constraints. However, microdiversity might not be suitable for combating large-scale fading, which is created by shadowing effects. Shadowing is almost independent of the frequency band, so that frequency diversity proves then ineffective. Spatial diversity can be used, but the correlation distances for large-scale fading can be greater than ten or even hundred meters. Thus, macrodiversity~\cite{Molisch2012}, where large distances between antennas exist, proves more appropriate to combat large-scale fading. 
	
Established principles to obtain independent fading channels include classical multiple-input multiple-output (MIMO) systems\cite{oestges2010} (spatial microdiversity) with space-time block coding, and distributed MIMO systems~\cite{oestges2010} (spatial macrodiversity), suitable for single-frequency networks \cite{Eriksson2001}. Alternatively, 4G concepts such as carrier aggregation (CA)~\cite{LTE_A} and dual connectivity (DC)~\cite{michalopoulos2016} have been introduced to make use of multiple so-called component carriers (frequency microdiversity), provided adjacent channels and sufficient RF bandwidth transceivers can realize frequency diversity with a single antenna. According to \cite{goldsmith2005}, the small-scale fading of two signals is approximately uncorrelated if their frequencies are separated at least by the coherence bandwidth, which is confirmed, for instance, by measurement results in \cite{van2012}. The techniques of CA and DC also support non-collocated deployments (frequency macrodiversity).
\subsection{Combining Algorithms}
\label{CombiningAlgorithms}
The system reliability strongly depends on the combining algorithm used at the receiver side, regardless of the diversity method. Combining algorithms merge the information received from multiple inputs (diversity branches) into a single unified output. The goal is to make use of the redundant information received from the multiple inputs. There are various combining algorithms known in literature \cite{tse2005}, many of which merge the received inputs at the symbol level, e.g., 
\begin{enumerate}
	\item \emph{Selection Combining}, where the best input in terms of received SNR is selected, while all other inputs are discarded, and
	\item \emph{Maximal-Ratio Combining}, where all received inputs are weighted by their respective SNRs and are coherently added.
\end{enumerate}	
In contrast to SC and MRC, which combine inputs already at the symbol level, the principle of
\begin{enumerate}
  	\setcounter{enumi}{2}
  	\item \emph{Joint Decoding}  is to combine the inputs at the decoder level, e.g., by iteratively exchanging information between decoders of each branch.
\end{enumerate}
 In addition to the different combining levels, JD differs fundamentally from SC and MRC in that the encoders at the transmitter may produce different channel codewords based on a joint codebook, while SC and MRC combine received inputs from the same channel codeword.
\subsection{Related Work}
\label{RelatedWork}
Recently, research on URLLC is emerging considerably, focusing on the analysis of micro- and macrodiversity and its impact on reliability. In~\cite{pocovi2015}, MCo solutions that utilize micro- as well as macrodiversity  were evaluated in system simulations to illustrate how the signal-to-interference-plus-noise ratio and the outage probability can be improved. In~\cite{Kirsten2015}, the impairments of correlated fading were evaluated and the trade-offs between power consumption, link usage, and outage probability were given. In another work, multi-radio access-technology architectures were compared regarding their latency, which is significantly improved by MCo techniques \cite{nielsen2016}. For other works on MCo for URLLC, see \cite{Oehemann2017} and the references therein. The major underlying concepts of MCo solutions, namely micro- and macrodiversity, have been extensively studied, and their effects on the outage probability are well understood~\cite{Molisch2012}. However, in the aforementioned studies, only linear (suboptimal) combining schemes, namely SC and MRC, have been considered. In particular, the JD scheme, which is optimum, remains open for investigation. Herein we help to fill this gap.

Deriving the outage probability of JD for parallel fading channels has been recognized as a highly challenging problem. An important result is to evaluate the diversity-multiplexing tradeoff (DMT) of fading channels. The DMT states that by doubling the SNR, we get both a decrease in outage probability by the factor of $2^{-d(r)}$, yielding an increase in reliability, and an increase in throughput of $r$ bits per channel use, i.e., the DMT describes the slope and the pre-log factor of the outage probability and throughput, respectively, at infinite SNR. This concept was first proposed by Zheng and Tse for MIMO channels~\cite{Zheng2003}. The corresponding results for parallel fading channels can be found in \cite{tse2005}. For finite-SNR the DMT of MIMO channels was proposed in \cite{Narasimhan2006} under correlated fading. However, the DMT analysis does not fully characterize the outage probability, and thus is not suitable to address the fundamental questions formulated at the beginning of the Introduction. In \cite{Bai2013},  a tight upper and lower bound on the outage probability based on the outage exponent analysis is given but it involves heavy computational efforts as the results include the incomplete Gamma function and Meijer's $G$-function. In addition, neither the DMT analysis in \cite{Zheng2003,tse2005} nor the outage exponent analysis in~\cite{Bai2013} considers macrodiversity.
\subsection{Asymptotic Outage Analysis}	
We want to find a good performance indicator to evaluate the reliability of MCo in light of URLLC applications. In fact, we aim to ultimately derive simple and insightful closed-form expressions that can be easily used to assess or optimize practical MCo deployments. To this end, an asymptotic analysis turns out to be a strong candidate, as it offers a simple yet in-depth characterization of the system performance's general trend. In the literature, the asymptotic outage probability is given depending on the so-called \emph{coding gain} $G_\text{C}$ and \emph{diversity gain}~$d$~(see, e.g., \cite{wang2003}), as
\begin{align*}
	\tilde{P}^{\text{out}}=\left(G_\text{C} \cdot \bar{\Gamma}\right)^{-d},
\end{align*}
where $\bar{\Gamma}$ is the average received SNR. In this work, we assume that the gain of MCo over SCo is based on the transmission of identical information over $N$ parallel block-fading channels. This setup can be exploited by all the combining algorithms described beforehand. All three combining algorithms can achieve the maximum diversity gain~\cite{tse2005}, i.e., $d=N$, whereas SCo has no diversity gain, i.e., $d=1$. But the combining algorithms differ with respect to the coding gain. The coding gains of SC and MRC have been studied in various contexts~\cite{duman2008}. However, the exact coding gain of JD has been unknown so far, since the derivation of the exact outage probability in closed form is very difficult. Only few bounds are known, e.g., a lower bound is given in \cite[Ch.~9.1.3]{tse2005} based on rate allocation to the individual fading channels, and lower and upper bounds based on an outage exponent analysis~\cite{Bai2013}. Both works reveal some drawbacks. The lower bound in~\cite[Ch.~9.1.3]{tse2005} offers a simple closed-form solution but is not tight, whereas the  lower and upper bounds in \cite{Bai2013} are tight but involve heavy computational efforts as the results include the incomplete Gamma function and Meijer's $G$-function. The exact solution of the coding gain of JD remains unknown and likewise the asymptotic outage probability.

More recently, we have evaluated the packet error rate of MCo by real link-level Monte-Carlo simulations in Wireless~LAN~\cite{Schwarzenberg2018}, from which we concluded that i) the asymptotic packet error rate is a good metric to evaluate the reliability of MCo for URLLC, ii) the asymptotic outage probability serves as a good benchmark to evaluate practical implementations, and iii) the asymptotic outage probability is suitable for link-level abstraction.
\subsection{Main Contributions of this Work}
\label{Contributions}
A convenient way to quantify the performance gain of MCo over SCo is to evaluate the required transmit power to achieve a given outage probability and a given spectral efficiency. Eventually, we are interested in the transmit power reduction of MCo with respect to SCo, which we refer to as \emph{SNR gain}. Based on the SNR gain we can answer the fundamental questions formulated at the beginning of the Introduction. To this end, we derive a remarkably simple analytical description of the asymptotic outage probability $\tilde{P}_{\text{JD},N}^{\text{out}}$ for JD depending on the number of links $N$, the spectral efficiency $R_\text{c}$, the average transmit SNRs per link $P_i/N_0$, and the path losses $d_i^{-\eta}$, for $i \in \{1,2,...,N\}$, as
\begin{align*}
\tilde{P}_{\text{JD},N}^{\text{out}} & = \frac{A_N(R_\text{c})}{\prod_{i=1}^{N} (P_i/N_0)d_i^{-\eta}} 
\intertext{where}
A_N(R_\text{c}) &  = (-1)^N\left(1 - 2^{R_\text{c}} \cdot e_N  \left(-R_\text{c} \ln(2)\right) \right),
\end{align*}
 with $e_N(\cdot)$ being the exponential sum function. The exponent $N$ of the SNR is the diversity gain and the $N$th root of the inverse numerator, i.e., $G_{\text{C},\text{JD}}=1/\sqrt[N]{A_N(R_\text{c})}$, is the coding gain. To the best of our knowledge, the exact coding gain of JD for parallel block-fading channels has been unknown so far. Our approach concentrates on the asymptotic solution, which is an accurate performance indicator for the operational region of URLLC applications, as we demonstrate by numerical examples. Based on the asymptotic outage probability, we derive the SNR gain of MCo over SCo as
\begin{align*}
G_{\text{MCo},\text{SCo}}=\frac{A_1(R_\text{c})}{N\sqrt[N]{A_N(R_\text{c})} } \frac{1}{\sqrt[N]{\left(P^\text{out}\right)^{N-1}}} \frac{\sqrt[N]{\prod_{i=1}^{N} d_i^{-\eta}}}{d_1^{-\eta}}.  
\end{align*}
This result reveals that the SNR gain of MCo over SCo increases at a rate of around $3(N-1)/N$~dB with respect to the target spectral efficiency (i.e., per source sample/channel symbol) and decreases at a rate of $4.3(N-1)/N\cdot1/P^\text{out}$~dB with respect to the target outage probability. In addition, we quantify the performance improvement of JD over SC and MRC in terms of the SNR gain as 
\begin{align*}
G_{\text{JD},\text{SC}} & =  \frac{A_1(R_\text{c})}{\sqrt[N]{A_N(R_\text{c})}}, \quad \text{and} \\
G_{\text{JD},\text{MRC}} & = \frac{1}{\sqrt[N]{N!}} \cdot \frac{A_1(R_\text{c})}{\sqrt[N]{A_N(R_\text{c})}}, 
\end{align*} 
respectively. These results reveal that the SNR gain of JD over SC and MRC increases at a rate of around $3(N-1)/N$~dB with respect to the target spectral efficiency (i.e., per source sample/channel symbol), while being insensitive to the target outage probability.\\ \indent 
Finally, we apply our analysis to real field channel measurements and thereby illustrate the potential of MCo in actual cellular networks to achieve high reliability and throughput.
\subsection{Notation and Terminology}
The upper- and lowercase letters are used to denote random variables (RVs) and their realizations, respectively, unless stated otherwise. The alphabet set of a RV $X$ with realization $x$ is denoted by $\mathcal{X}$, and its cardinality, by $|\mathcal{X}|$. The probability mass function (pmf) and probability density function (pdf) of the discrete and continuous RV $X$ is denoted by $p_X(x)$ and $f_X(x)$, respectively. The pmf and pdf are simply denoted by $p(x)$ and $f(x)$, respectively, whenever this notation is unambiguous. Also, $\mathbf{X}^n$ and $\mathbf{x}^n$ represent vectors containing a temporal sequence of $X$ and $x$ with length $n$, respectively. We use $t$ to denote the time index and $i$ to denote a source index. We define $\mathbf{A}_\mathcal{S} = \{\mathbf{A}_i | i \in \mathcal{S}\}$ as an indexed series of random vectors, and $A_\mathcal{S} = \{A_i | i \in \mathcal{S}\}$ as an indexed series of RVs. In general, a set $\mathbf{A}$ contains elements $a_{(\cdot)}$, as in $\mathbf{A}=\{a_1,a_2,...,a_{|\mathbf{A}|}\}$. We define one particular set: $\mathcal{N}=\{1,2,...,N\}$. For two function $f(x)$ and $g(x)$, the notation $f(x)=\Theta(g(x))$, means that $k_1 g(x) \leq f(x) \leq k_2 g(x), \exists k_1>0,\exists k_2>0,\exists x_0,\forall x>x_0 $. We denote the probability of an event $\mathcal{E}$ as $\text{Pr}[\mathcal{E}]$, the mutual information as $I(\cdot;\cdot)$, the entropy as $H(\cdot)$, the convolution as $\ast$, the binary logarithm as $\ld(\cdot)$, and the natural logarithm as $\ln(\cdot)$.
\section{System Model}
\label{SystemModel}
\begin{figure*}[!ht]
	\centering
	\begin{subfigure}{0.45\linewidth}
		\includegraphics[width=\linewidth]{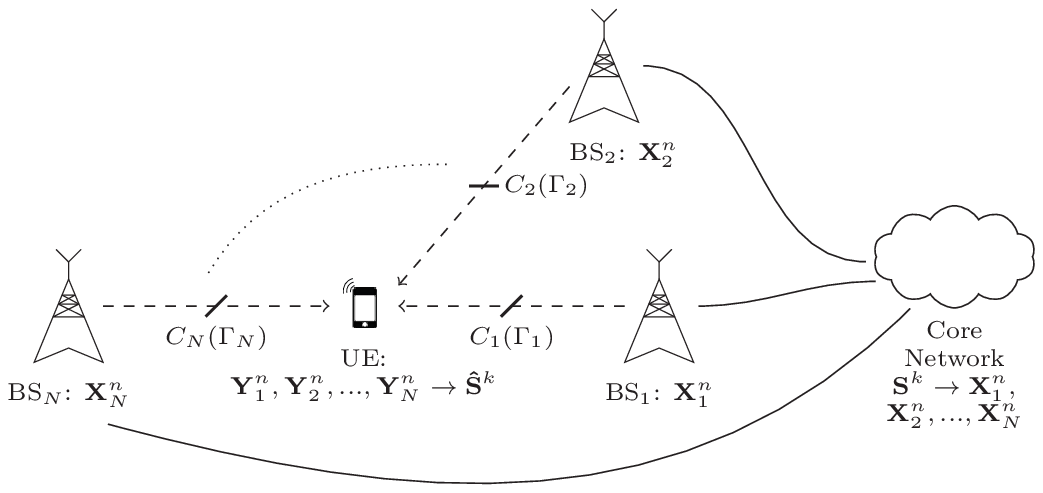}
		\caption{}
		\label{fig:SystemModel_DL}
	\end{subfigure}
	\begin{subfigure}{0.45\linewidth}
		\includegraphics[width=\linewidth]{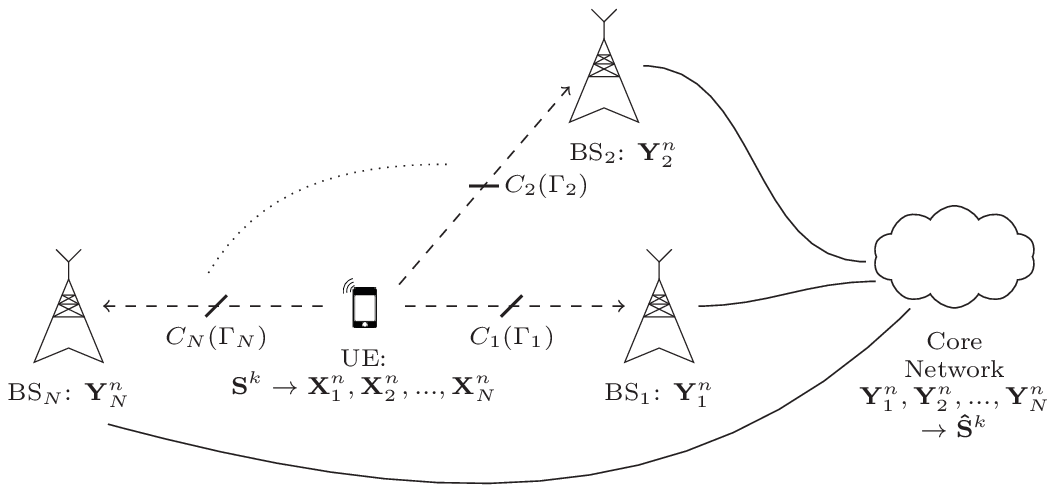}
		\caption{}
		\label{fig:SystemModel_UL}
	\end{subfigure}
	\caption{System model with a single UE, $N$ base stations, and a core network for (a) the downlink and (b) the uplink.}
	\label{fig:SystemModel_DLUL}
\end{figure*}
\subsection{Multi-Connectivity System Model}
\label{MCoSystemModel}
We consider a MCo cellular network consisting of a core network and $N$ base stations (BS$_i$, $\forall i \in \mathcal{N}$) communicating to a single user equipment (UE). The core network coordinates the data transmissions to the UE. To ease the notation, only one link per BS is included in the system model (macrodiversity). However, a single BS can also establish multiple connections to the UE (microdiversity). Connections between the core network and each BS are realized by backhaul links, and connections between each BS and UE, by wireless links. By assumption all $N$ wireless links are orthogonal, i.e., we consider parallel block-fading channels. The achievable transmission rate over the $i$th wireless link depends on its capacity $C_i(\Gamma_i)$ and, thus, on its received SNR $\Gamma_i$. In the system model, we distinguish between down- and uplink as follows: 
\paragraph{Downlink} The downlink system model, as illustrated in Fig.~\ref{fig:SystemModel_DL}, has one binary memoryless source, denoted as $[S(t)]_{t=1}^\infty$, with the $k$-sample sequence being represented in vector form as $\mathbf{S}^k = [S(1),S(2),...,S(k)]$. When appropriate, for simplicity, we shall drop the temporal index of the sequence, denoting the source merely as $S$. By assumption $S$ takes values in a binary set $\mathcal{B}=\{0,1\}$ with uniform probabilities, i.e., $\text{Pr}[S=0]=\text{Pr}[S=1]=0.5$. Therefore, the entropy of the sequence is $\sfrac{1}{k} \cdot H(\mathbf{S}^k)=H(S)=1$. At the core network, the source sequence $\mathbf{S}^k$ is encoded to the transmit sequences $\mathbf{X}_i^n, \forall i \in \mathcal{N}$. The $i$th transmit sequence is forwarded over the backhaul link to BS$_i$. At each BS the transmit sequence is then sent to the UE over parallel block-fading channels. The decoder at the UE retrieves the source sequence $\mathbf{S}^k$ from the received sequences $\mathbf{Y}_i^n, \forall i \in \mathcal{N}$.
\paragraph{Uplink} The uplink system model, as illustrated in Fig.~\ref{fig:SystemModel_UL}, is similar to the downlink, except that the source sequence is originated from the UE, and the received sequences $\mathbf{Y}_i^n, \forall i \in \mathcal{N}$, are decoded at the core network to retrieve the source sequence $\mathbf{S}^k$. Similar to the downlink, the $i$th transmit sequence $\mathbf{X}_i^n$ is sent from the UE to the $i$th BS  over parallel block-fading channels.\\ \indent
In this work, the down- and uplink system models can be treated interchangeably, so that all further results are applicable to both system models. 

\subsection{Link Model}
\label{LinkModel}
As discussed in the Introduction, micro- and macrodiversity can be achieved by spatial or frequency separation of the channels. In both cases, it is reasonable to assume that the channel fading is approximately uncorrelated, which we do in this study. Furthermore, to cope with the low latency constraint in URLLC, we consider relatively short encoded sequences. As a result, the length of an encoded sequence is less than or equal to the length of a fading block of a block Rayleigh fading. Moreover, the signals can be transmitted from or to different BSs, which leads to individual average SNR values.  \\ \indent
As argued, we can assume that the sequences $\mathbf{X}^n_i$, $\forall i \in \mathcal{N}$, are transmitted (in up- and downlink) over independent channels undergoing block Rayleigh fading and additive white Gaussian noise (AWGN) with mean power $N_0$. The pdf of the received SNR $\Gamma_i$ is given~by
\begin{gather}
\label{eq:RFpdf}
f_{\Gamma_i}(\gamma_i)=\frac{1}{\bar{\Gamma}_i} \exp\left(-\frac{\gamma_i}{\bar{\Gamma}_i}\right), \quad \text{for } \gamma_i \geq 0,
\end{gather}
with the average SNR $\bar{\Gamma}_i$ being obtained as $\bar{\Gamma}_i=(P_i/N_0) \cdot d_i ^{-\eta},$ where $P_i$ is the transmit power per channel, $d_i$ is the distance between BS$_i$ and the UE, and $\eta$ is the path loss exponent. The channel state information is assumed to be exclusively known at the receiver.
\section{Preliminaries and Problem Statement}
\label{Preliminaries}

The problem investigated in this work can be formulated as follows: the received channel codewords $\mathbf{Y}_1^n,...,\mathbf{Y}_N^n$ (cf. received sequences in Section~\ref{MCoSystemModel}) must comprise sufficient information such that the source sequence $\mathbf{S}^k$  can be perfectly reconstructed. At the core network, for the $i$th link, a channel code maps the source sequence $\mathbf{S}^k$ to a channel codeword $\mathbf{X}_i^n$ (cf. transmitted sequence in Section~\ref{MCoSystemModel}) with an spectral efficiency $R_{i,\text{c}}$, measured in source samples per channel input symbol, associated with  the modulation scheme $R_{i,\text{M}}=\ld(M)$, measured in bits per channel input symbol for a cardinality $M$ of the channel input symbol alphabet; and the channel code rate $R_{i,\text{cod}}$, measured in source samples per bit; i.e., $R_{i,\text{c}}=R_{i,\text{M}} \cdot R_{i,\text{cod}}$. In many parts of this work, for simplicity, we shall assume  $R_{i,\text{c}}= R_{\text{c}}=k/n, \forall i$. As mentioned before, the transmitter has no knowledge of the CSI and chooses a fixed spectral efficiency. If and only if the instantaneous capacity of the parallel fading channels support the spectral efficiency, the error probability at the receiver side can be made arbitrarily small~\cite{Biglieri1998}. 
\subsection{Channel Capacity}
The instantaneous channel capacity of the MCo system model depends on the combining algorithm. For SC and MRC the analysis is based on the gain attained by receiving the same channel codeword via multiple links, i.e., the channel codewords are identical: $\mathbf{X}_1^n=...=\mathbf{X}_N^n=\mathbf{X}^n$. On the other hand, for JD the channel codewords transmitted over all links may be different but based on a joint codebook. The resulting instantaneous channel capacities are given as follows: \\ \indent
\emph{Selection Combining:} The instantaneous received SNRs of each channel are assumed to be known by the receiver. The received sequence of the subchannel with the maximum SNR is selected as combiner output, and the other received sequences are discarded, so that~\cite{tse2005}
\begin{align}
	C_\text{SC}\left(\Gamma_1,...,\Gamma_N\right) & =  \max\left(I(\mathbf{X}^n;\mathbf{Y}_1^n),...,I(\mathbf{X}^n;\mathbf{Y}_N^n)\right) \nonumber \\
	\label{eq:capacity_SC}
	& =  \phi\left(\max\left(\Gamma_1,...,\Gamma_N\right)\right)
\end{align} 
is the complex AWGN channel capacity of SC. Here and throughout the text, $\phi(x)=\ld(1+x)$ represents the instantaneous complex AWGN channel capacity.\\ \indent
\emph{Maximal-ratio combining:} As for SC, the instantaneous received SNRs of each subchannel is assumed to be known. The received sequences are scaled to their corresponding SNRs and coherently added, so that~\cite{tse2005}
\begin{align}
	C_\text{MRC}\left(\Gamma_1,...,\Gamma_N\right)& =  I(\mathbf{X}^n;\mathbf{Y}_1^n,...,\mathbf{Y}_L^n)\nonumber \\
\label{eq:capacity_MRC}
	& =  \phi\left(\sum\nolimits_{i=1}^N\Gamma_i\right)
\end{align}
is the complex AWGN channel capacity of MRC. \\ \indent
\emph{Joint Decoding:} In contrast to SC and MRC, JD processes each received sequence individually until decoding. The achieved complex AWGN channel capacity is~\cite{Biglieri1998}
\begin{align}
C_\text{JD}\left(\Gamma_1,...,\Gamma_N\right) & =I(\mathbf{X}_1^n;\mathbf{Y}_1^n)+...+I(\mathbf{X}_N^n;\mathbf{Y}_N^n)\nonumber \\
\label{eq:capacity_JD}
& =\sum\nolimits_{i=1}^N\phi\left(\Gamma_i\right).
\end{align}
\subsection{Outage Formulation}
The outage probability is an important concept in fading channels, which provides a way to characterize the performance of communication systems in non-ergodic fading scenarios. As $N$ is finite, the parallel fading channel is non-ergodic, i.e., $N$ is not large enough to average over channel variations. If the achieved instantaneous channel capacity is less than the spectral efficiency, an outage event occurs~\cite{tse2005}. Thus, the outage probability is given by
\begin{align}
\label{eq:outage}
	P^\text{out}_{j,N}=\text{Pr} \left[C_j\left(\Gamma_1,...,\Gamma_N\right) < R_\text{c}\right]
\end{align}
for $j \in \{\text{SC},\text{MRC},\text{JD}\}$. The outage probability for SC and MRC are well studied, whereas it is very difficult to derive a tractable exact formula of the outage probability of JD.

\section{Outage Probabilities}
\label{Outage}
In this section, we establish the exact outage probability of JD in integral form. More importantly, we derive in closed form a corresponding asymptotic expression for the high-SNR regime. In addition, we reproduce known bounds on the JD outage probability given in \cite{tse2005,Bai2013}, as well as the exact and asymptotic outage probabilities of SC and MRC~\cite{duman2008,chen2005}, which we require later on for comparison.
\subsection{Joint Decoding}
The outage probability for JD can be calculated as follows:
{\allowdisplaybreaks\begin{align}
	\label{eq:JD_Outage_1}
	& P_{\text{JD},N}^{\text{out}} =\text{Pr}\left[0 \leq \phi(\Gamma_1) + \phi(\Gamma_2) + ... + \phi(\Gamma_N) < R_\text{c} \right] \\
				 & =\text{Pr}\left[0 \leq \phi(\Gamma_1) < R_\text{c} , 0 \leq \phi(\Gamma_2) < R_\text{c} - \phi(\Gamma_1),..., \right.\nonumber \\ 
	\label{eq:JD_Outage_2}
				 & \left. \qquad \,\, 0 \leq \phi(\Gamma_N) < R_\text{c}  - \phi(\Gamma_1) - \phi(\Gamma_2) - ... -  \phi(\Gamma_{N-1})\right] \\
				 & =\text{Pr}\left[0 \leq \Gamma_1 < 2^{R_\text{c} }-1, 0 \leq \Gamma_2<  2^{ R_\text{c}  - \phi(\Gamma_1)}-1,..., \right.\nonumber \\ 
 \label{eq:JD_Outage_3}
				 & \left. \qquad \,\, 0 \leq \Gamma_N < 2^{ R_\text{c}  - \phi(\Gamma_1) - \phi(\Gamma_2) - ... -  \phi(\Gamma_{N-1})}-1\right] \\
				 & = \int_{\gamma_1=0}^{2^{R_\text{c} }-1} \int_{\gamma_2=0}^{2^{ R_\text{c}  - \phi(\gamma_1)}-1} ... \int_{\gamma_N=0}^{2^{ R_\text{c}  - \phi(\gamma_1) - \phi(\gamma_2) - ... -  \phi(\gamma_{N-1})}-1}  \nonumber \\ 
	\label{eq:JD_Outage_4}
					& \quad f(\gamma_1) f(\gamma_2) ... f(\gamma_N) \text{d} \gamma_N ... \text{d} \gamma_2 \text{d} \gamma_1.
\end{align}}\noindent
The steps are justified as follows: in (\ref{eq:JD_Outage_1}) we substitute the JD channel capacity in~(\ref{eq:capacity_JD})  into the outage probability expression in~(\ref{eq:outage}); in (\ref{eq:JD_Outage_2}) the sum constraint is separated into individual constraints; in (\ref{eq:JD_Outage_3}) the bounds are transformed with $\phi^{-1}(y)=2^y-1$; in (\ref{eq:JD_Outage_4}) the probability of outage is established in integral form with the assumption that the received SNRs $\Gamma_i, \forall i \in \mathcal{N}$, are independent. The pdf $f(\gamma_i)$ is given in (\ref{eq:RFpdf}). Although the outage expression in (\ref{eq:JD_Outage_4}) cannot be solved in closed form, a simple asymptotic solution can be derived at high~SNR as
{\allowdisplaybreaks\begin{align}
		P_{\text{JD},N}^{\text{out}} 		 & \approx \int_{\gamma_1=0}^{2^{R_\text{c} }-1} \int_{\gamma_2=0}^{2^{ R_\text{c}  - \phi(\gamma_1)}-1} ... \noindent \nonumber \\
		& \quad \int_{\gamma_N=0}^{2^{ R_\text{c}  - \phi(\gamma_1) - \phi(\gamma_2) - ... -  \phi(\gamma_{N-1})}-1} \nonumber \\
\label{eq:JD_Outage_5}
		& \quad \frac{1}{\prod_{i=1}^{N}\bar{\Gamma}_i} \text{d} \gamma_N ... \text{d} \gamma_2 \text{d} \gamma_1 \\
	\label{eq:JD_Outage_6}
				 &= \frac{A_N(R_\text{c})}{\prod_{i=1}^{N}\bar{\Gamma}_i} \quad \text{where} \\
\label{eq:Constant_JD}
	A_N(R_\text{c}) & = (-1)^N\left(1 - 2^{R_\text{c}} \cdot e_N  \left(-R_\text{c} \ln(2)\right) \right).
\end{align}}\noindent
Here, $e_N(x)=\sum_{n=0}^{N-1}\frac{x^n}{n!}$ is the exponential sum function. For more details, we refer to the derivations in Appendix~\ref{app:ExactOutage}. The $N$th root of the inverse numerator $G_{\text{C},\text{JD}}=1/\sqrt[N]{A_N(R_\text{c})}$ is commonly termed the coding gain~\cite{wang2003}.  To our best knowledge, the coding gain of JD was unknown so far and constitutes an important original contribution of this work.\\ \indent
In contrast to the asymptotic solution in (\ref{eq:JD_Outage_6}), a lower bound for (\ref{eq:JD_Outage_4}) is given by~\cite[Ch.~9.1.3]{tse2005}
	\begin{align}
	P_{\text{JD},N}^{\text{out}} &  \geq 	\left[\text{Pr}[0 \leq \phi(\Gamma) < R_\text{c}/N]\right]^N \nonumber \\
	\label{eq:JD_Outage_7}
	& = \left[1-\exp\left(- \frac{A_1(R_\text{c}/N)}{\bar{\Gamma}}\right)\right]^N,
	\end{align}
for $\bar{\Gamma}_1=...=\bar{\Gamma}_N=\bar{\Gamma}$ and $A_1(R_\text{c}/N)=2^{R_\text{c}/N}-1$. Note that the lower bound is based on the assumption that an outage event occurs if the channel capacity of at least one fading channel cannot support the spectral efficiency $R_\text{c}/N$, i.e., each fading channel is allocated an equal share of the information. \\ \indent
In~\cite{Bai2013}, based on the outage exponent analysis, a lower and an upper bound are given by
\begin{equation}
\label{eq:JD_Outage_8}
P_{\text{JD},N}^{\text{out}} 
\begin{cases}
\geq & P_{\text{JD},N}^{\text{out,lower}}=a \exp\big(N \big[\left(\phi\left(\bar{\Gamma}\right)-\sfrac{R_\text{c}}{N}\right) E_{1,1}\left(\bar{\Gamma}\right) \\
& \qquad \qquad \,\, + E_{1,0}\left(\bar{\Gamma}\right)+\sfrac{E_0(\bar{\Gamma})}{N}+ o(N)\big]\big),\\
\leq & P_{\text{JD},N}^{\text{out,upper}}=b \exp\big(N \big[\left(\phi\left(\bar{\Gamma}\right)-\sfrac{R_\text{c}}{N}\right) E_{1,1}\left(\bar{\Gamma}\right)  \\
& \qquad \qquad \,\,+ E_{1,0}\left(\bar{\Gamma}\right)+\sfrac{E_0(\bar{\Gamma})}{N}+ o(N)\big]\big)
\end{cases}       
\end{equation}
where $a$ and $b$ are constants, with $a \leq b$. $P_{\text{JD},N}^{\text{out,lower}}$ and $P_{\text{JD},N}^{\text{out,upper}}$ are refereed to as the lower and upper outage exponents, respectively. The exact reliability functions $E_{1,1}\left(\bar{\Gamma}\right)$, $E_{1,0}\left(\bar{\Gamma}\right) $, and $ E_{0}\left(\bar{\Gamma}\right) $ are given in~\cite{Bai2013}. According to~\cite{Bai2013}, the outage probability differs for $R_\text{c}/N < C_\text{ergodic}$ and  $R_\text{c}/N \geq C_\text{ergodic}$, where $C_\text{ergodic} = \lim\nolimits_{N \rightarrow \infty} C_\text{JD}(\Gamma_1,...,\Gamma_N)/N = \mathds{E}\left[\phi(1+\Gamma)\right]$
is the ergodic capacity. The derivations of (\ref{eq:JD_Outage_8}) are mainly based on large deviations theory and Meijer's $G$-function~\cite{gradshteyn2014,dembo2010}. For more details on the outage exponent analysis, please refer to~\cite{Bai2013}. \\ \indent
In contrast to the bounds in (\ref{eq:JD_Outage_7}) and (\ref{eq:JD_Outage_8}), our asymptotic solution in (\ref{eq:JD_Outage_6}) offers a simple, yet accurate solution at high SNR. The outage exponent analysis in (\ref{eq:JD_Outage_8}) can achieve tight bounds on the outage probability, but the calculations involve  the incomplete Gamma function and Meijer's $G$-function, which makes any further analytical derivations on the SNR gain, DMT, and throughput rather involved. As \cite{Bai2013} did not explicitly considered the properties of the asymptotic outage probability, the exact solution of the coding gain remained unknown. The lower bound in (\ref{eq:JD_Outage_7}) is a simple solution, but not tight, as we show later on. On the other hand, our asymptotic solution in (\ref{eq:JD_Outage_6}) offers a remarkable simple asymptotical description of the outage probability at high SNR. Especially for URLLC, high-SNR results are well suited, as we are interested in frame error rates below $10^{-5}$. Note that, more generally than (\ref{eq:JD_Outage_7}) and (\ref{eq:JD_Outage_8}), our solution also allows for different average SNRs, which is of practical relevance, if  the signals are transmitted from or to different BSs. We detail this comparison via numerical examples in Section~\ref{NumericalResults}.
\subsection{Linear Combining}
\label{SC}
The outage probability for SC can be derived as follows~\cite[(2.42)]{duman2008}:
{\allowdisplaybreaks\begin{align}
		\label{eq:SC_Outage_1}
		P_{\text{SC},N}^{\text{out}} & = \text{Pr}\left[0 \leq \phi(\max\left(\Gamma_1,...,\Gamma_N\right) ) < R_\text{c} \right] \\
		\label{eq:SC_Outage_2}
		& = \prod\nolimits_{i=1}^N\int_{\gamma_i=0}^{2^{R_\text{c} }-1} f(\gamma_i) \text{d} \gamma_i \\
		\label{eq:SC_Outage_3}
		& = \prod\nolimits_{i=1}^{N}\bigg(1- \exp \bigg(-\frac{A_1(R_\text{c})}{\bar{\Gamma}_i}\bigg) \bigg), 
		\end{align}}\noindent
with $A_1(R_\text{c})=2^{R_\text{c}}-1$. The steps are justified as follows: in (\ref{eq:SC_Outage_1}) we substitute the SC channel capacity in~(\ref{eq:capacity_SC}) into the outage probability expression in~(\ref{eq:outage}); (\ref{eq:SC_Outage_2}) follows similar arguments as in (\ref{eq:JD_Outage_1}), (\ref{eq:JD_Outage_3}), and (\ref{eq:JD_Outage_4}), respectively, and the multiple integral can be rewritten as the product of single integrals, since the integral domain is normal and the SNRs are independent; (\ref{eq:SC_Outage_3}) is the closed-form solution of the integral in (\ref{eq:SC_Outage_2}). An asymptotic solution at high SNR can be derived by using the MacLaurin series of the exponential function  $\exp(-x_i) \approx 1- x_i$ for $x_i\rightarrow 0$, giving \cite[(2.43)]{duman2008}
\begin{align}
\label{eq:SC_Outage_4}
P_{\text{SC},N}^{\text{out}} & \approx  \frac{\left(A_1(R_\text{c})\right)^N}{\prod_{i=1}^{N}\bar{\Gamma}_i}, 
\end{align}
from which $G_{\text{C},\text{SC}}=1/A_1(R_\text{c})$ is the coding gain of SC. \\ \indent
To calculate the outage probability of MRC we define an auxiliary RV, namely, the total received SNR, as $\Gamma_\text{MRC} = \sum\nolimits_{i=1}^{N} \Gamma_i$.\\ \indent
We have to distinguish between two cases:
\paragraph{Identical received average SNRs} For $\bar{\Gamma}_1=...=\bar{\Gamma}_N=\bar{\Gamma}$ the pdf of the total received SNR is given by~\cite[(2.30)]{duman2008}
	\begin{align}
	\label{eq:pdfMRA_1}
	f_{\Gamma_\text{MRC}}(\gamma_\text{MRC})= \frac{\gamma_\text{MRC}^{(N-1)}}{(N-1)! \cdot \bar{\Gamma}^N}  \exp\left(-\frac{\gamma_\text{MRC}}{\bar{\Gamma}}\right).
	\end{align}
The outage probability can be then calculated as follows~\cite[(2.31)-(2.33)]{duman2008}:
	{\allowdisplaybreaks	\begin{align}
	\label{eq:MRC_Outage_1}
	P_{\text{MRC},N}^{\text{out}} & = \text{Pr}\left[0 \leq \phi(\Gamma_\text{MRC} ) < R_\text{c} \right] \\
	\label{eq:MRC_Outage_2}
	& = \int_{\gamma_\text{MRC}=0}^{2^{R_\text{c} }-1} f(\gamma_\text{MRC}) \text{d} \gamma_\text{MRC} \\
	\label{eq:MRC_Outage_3}
	& = 1 - \exp\bigg(-\frac{A_1(R_\text{c})}{\bar{\Gamma}}\bigg) \bigg(\sum\nolimits_{i=1}^{N} \frac{\left(\sfrac{A_1(R_\text{c})}{\bar{\Gamma}}\right)^{(i-1)}}{(i-1)!}\bigg). 
	\end{align}}\noindent
The steps can be justified similarly to (\ref{eq:SC_Outage_1})-(\ref{eq:SC_Outage_2}). The closed form of the integral in (\ref{eq:MRC_Outage_2}) is given in \cite[(2.33)]{duman2008}. An asymptotic solution can be derived at high SNR \cite[(16)]{chen2005} as
\begin{align}
\label{eq:MRC_Outage_4}
P_{\text{MRC},N}^{\text{out}} & \approx \frac{1}{N!}\bigg(\frac{A_1(R_\text{c})}{\bar{\Gamma}}\bigg)^N,
\end{align}
from which $G_{\text{C},\text{MRC}}=\sqrt[N]{N!}/A_1(R_\text{c})$ is the coding gain of MRC.
\paragraph{Different received average SNRs} For $\bar{\Gamma}_1\neq...\neq\bar{\Gamma}_N$ the pdf of the total received SNR is given by~\cite[Proposition~3.1]{Bibinger2013}
\begin{align}
	f_{\Gamma_\text{MRC}}(\gamma_\text{MRC})& = \sum\nolimits_{i=1}^{N} \bar{\Gamma}_i^{N-2} \exp \left(-\frac{\gamma_\text{MRC}}{\bar{\Gamma}_i}\right) \nonumber\\
	\label{eq:pdfMRA_2}
	& \quad \times \prod\nolimits_{\substack{j=1\\j \neq i}}^{N} \frac{1}{\bar{\Gamma}_i-\bar{\Gamma}_j} .
\end{align}
The outage probability can be then calculated as follows:
	{\allowdisplaybreaks	\begin{align}
		\label{eq:MRC_Outage_5}
		P_{\text{MRC},N}^{\text{out}} & = \text{Pr}\left[0 \leq \phi(\Gamma_\text{MRC} ) < R_\text{c} \right] \\
 		& = \int_{\gamma_\text{MRC}=0}^{2^{R_\text{c} }-1} \sum\nolimits_{i=1}^{N} \bar{\Gamma}_i^{N-2} \exp \left(-\frac{\gamma_\text{MRC}}{\bar{\Gamma}_i}\right)  \nonumber\\
 		\label{eq:MRC_Outage_8}
 		& \quad \times  \prod\nolimits_{\substack{j=1\\j \neq i}}^{N} \frac{1}{\bar{\Gamma}_i-\bar{\Gamma}_j} \text{d} \gamma_\text{MRC} \\
 		& = \sum\nolimits_{i=1}^{N} \bar{\Gamma}_i^{N-1} \left( 1-  \exp \left(-\frac{A_1(R_\text{c})}{\bar{\Gamma}_i}\right) \right)  \nonumber\\
 		\label{eq:MRC_Outage_9}
 		& \quad \times  \prod\nolimits_{\substack{j=1\\j \neq i}}^{N} \frac{1}{\bar{\Gamma}_i-\bar{\Gamma}_j}.
		\end{align}}\noindent
The steps can be justified similarly to (\ref{eq:SC_Outage_1})-(\ref{eq:SC_Outage_3}). No further simplification based on high SNRs can be achieved for (\ref{eq:MRC_Outage_9}). However, the asymptotic solution of identical received average SNRs in (\ref{eq:MRC_Outage_4}) gives an upper bound at high SNR for (\ref{eq:MRC_Outage_9}) by replacing $\bar{\Gamma}^N$ with $\prod_{i=1}^{N}\bar{\Gamma}_i$ yielding
\begin{align}
 		\label{eq:MRC_Outage_10}
 		P_{\text{MRC},N}^{\text{out}} & \lesssim  \frac{1}{N!}\frac{\left(A_1(R_\text{c})\right)^N}{\prod_{i=1}^{N}\bar{\Gamma}_i}. 
\end{align}
 For more details on this upper bound, we refer to the derivations in Appendix~\ref{app:UpperBound}.
\subsection{Single-Connectivity}
In addition, the exact and asymptotic outage probability of SCo (e.g., $N=1$ for (\ref{eq:SC_Outage_3}) and (\ref{eq:SC_Outage_4})) are given as a baseline by
\begin{align}
\label{eq:SCo_Outage_1}
	P_\text{SCo}^\text{out} & = 1- \exp \left(-\frac{A_1(R_\text{c})}{\bar{\Gamma}} \right) \\
\label{eq:SCo_Outage_2}	
	& \approx \frac{A_1(R_\text{c})}{\bar{\Gamma}},
\end{align}
from which $G_{\text{C},\text{SCo}}=1/A_1(R_\text{c})$ is the coding gain of SCo. 
\section{Throughput}
\label{Throughput}
The throughput captures how much information is received at the destination on average per transmission, depending on the SNR. To capture this, following the standard approach in the literature \cite{Molisch2012}, we define the throughput $T$ as the product of the bandwidth $B$, spectral efficiency $R_\text{c}$, and the non-outage probability $(1-P^\text{out})$, i.e.,
\begin{align}
	\label{eq:system_throughput}
	T= B  R_\text{c} (1-P^\text{out}) \text{ in bit/s}.
\end{align}
To evaluate (\ref{eq:system_throughput}), we require the achieved spectral efficiency $R_\text{c}$ for the different combining algorithms ($j\in \{\text{JD},\text{SC},\text{MRC}\}$) for a given number of links $N$, outage probability $P^\text{out}$, and average received SNRs $\bar{\Gamma}_\mathcal{N}$.\\ \indent  
Now, by reformulating (\ref{eq:JD_Outage_6}) so as to express the achieved spectral efficiency in terms of a given outage probability, we obtain a high-SNR asymptotic expression for the throughput of JD as
 \begin{align}
 \label{eq:system_throughput_JD}
 T_{\text{JD},N} & \approx B  A_N^{-1}\left(P^\text{out} \prod\nolimits_{i=1}^{N}\bar{\Gamma}_i\right)(1-P^\text{out})  \text{ in bit/s}.
\end{align}
Here, $A_N^{-1}(\cdot)$ is the inverse function of $A_N(\cdot)$. Unfortunately, the inverse function $A_N^{-1}(\cdot)$ does not have a closed-form solution. However, we can give a good approximation. At high SNR, the achieved spectral efficiency for a given outage probability is $R_\text{c} \gg 1$. In this case, the inverse function can be given by use of an approximation of the asymptotic Lamber $W$ function \cite{hoorfar2007}~by
\begin{align}
 \label{eq:Inverse_Function}
A_N^{-1}\left(P^\text{out} \prod\nolimits_{i=1}^{N}\bar{\Gamma}_i\right)& =R_\text{c} \approx \frac{N-1}{\ln(2)} \left[\ln(\zeta)-\ln(\ln(\zeta))\right],  
\intertext{where}
\zeta & =\frac{\sqrt[N-1]{ (N-1)! P^\text{out} \prod\nolimits_{i=1}^{N}\bar{\Gamma}_i}}{N-1}.
\end{align}
 For more details, we refer to the derivations in Appendix~\ref{app:InverseFunction}. The asymptotic throughputs for SC and MRC can be given based on (\ref{eq:SC_Outage_4}) and (\ref{eq:MRC_Outage_10}) as
 \begin{align}
	T_{\text{SC},N} & \approx B \ld\left( \sqrt[N]{P^\text{out} \prod\nolimits_{i=1}^{N}\bar{\Gamma}_i} +1 \right) \nonumber \\
  	\label{eq:system_throughput_SC}
  	& \quad \times 	(1-P^\text{out}) \text{ in bit/s} , \\
	T_{\text{MRC},N} & \approx B\ld\left( \sqrt[N]{N! \cdot P^\text{out} \prod\nolimits_{i=1}^{N}\bar{\Gamma}_i} +1 \right) \nonumber \\
	\label{eq:system_throughput_MRC}
	& \quad \times	(1-P^\text{out}) \text{ in bit/s},
\end{align}
respectively. In addition, we give the asymptotic throughput of SCo (e.g., from $N=1$ in (\ref{eq:system_throughput_SC})):
\begin{align}
  \label{eq:system_throughput_SCo}
	T_{\text{SCo}} & \approx B \ld\left( P^\text{out} \bar{\Gamma} +1 \right)(1-P^\text{out})  \text{ in bit/s}.
\end{align}
\section{Multi-Connectivity Gain}
\label{MCo_Gain}
In this section, we quantify the performance gain of MCo over SCo in terms of transmit power reduction. For MCo we consider the optimal combining algorithm, i.e., JD. To ensure a fair comparison between different setups, we equally allocate the total transmit power $P_\text{T}$ to all channels, such that $P_i=P_\text{T}/N, \forall i \in \mathcal{N}$. However, this assumption is non-essential, and other system setups can be evaluated from our formulas with some effort. \\ \indent
We assume a target (fixed) spectral efficiency $R_\text{c}$ (i.e., a certain throughput has to be guaranteed) and target (fixed) outage probability $P^{\text{out}}$, and evaluate the required total average SNR $\bar{\Gamma} =\sigma_{(\cdot)}(P^\text{out})$ in the high-SNR regime. The SNR gain is defined as the ratio of the required average SNR between SCo and MCo.\\ \indent
A reformulation of (\ref{eq:JD_Outage_6}) yields
\begin{align}
\label{eq:reqSNR_JD}
\sigma_{\text{JD}}(P^{\text{out}})=\frac{P_\text{T}}{N_0}= N \sqrt[N]{\frac{A_N(R_\text{c})}{P^{\text{out}}}}\frac{1}{\sqrt[N]{\prod_{i=1}^{N} d_i^{-\eta}}},
\end{align}
where $\sigma_{\text{JD}}(P^{\text{out}})$ is the required total average SNR for JD. The required total average SNR for SCo based on the reformulation of (\ref{eq:SCo_Outage_2}) is 
\begin{align}
\label{eq:reqSNR_SCo}
\sigma_{\text{SCo}}(P^{\text{out}})=\frac{P_\text{T}}{N_0}=  \frac{A_1(R_\text{c})}{P^\text{out}} \frac{1}{d_1^{-\eta}}.
\end{align}
The SNR gain is then given as the ratio between the required average SNRs for SCo and JD as
\begin{align}
G_{\text{MCo},\text{SCo}}& =\frac{\sigma_{\text{SCo}}(P^{\text{out}})}{\sigma_{\text{JD}}(P^{\text{out}})} \nonumber \\
\label{eq:SNRGain_JD}
& =\frac{A_1(R_\text{c})}{N\sqrt[N]{A_N(R_\text{c})} } \frac{1}{\sqrt[N]{\left(P^\text{out}\right)^{N-1}}}  \frac{\sqrt[N]{\prod_{i=1}^{N} d_i^{-\eta}}}{d_1^{-\eta}}.
\end{align}
Based on (\ref{eq:SNRGain_JD}), we can answer the fundamental questions formulated at the beginning of the Introduction, i.e., how much transmit power can be saved by MCo as compared with SCo depending on the number of links $N$, the spectral efficiency  (corresponding to the throughput $T \approx B R_\text{c}$), the path loss $d_i^{-\eta}$ for $i \in \mathcal{N}$, and the outage probability $P^\text{out}$. 
\section{Joint Decoding vs. Linear Combining}
In this section, we evaluate the performance improvement of JD over SC and MRC. All combining algorithms for MCo are superior to SCo, since the multiple diversity branches are exploited, i.e., the diversity gain is $N$. However, there exists a difference of the outage probabilities governed by the coding gains $G_{\text{C},(\cdot)}$ of each combining algorithm ((\ref{eq:JD_Outage_6}), (\ref{eq:SC_Outage_4}), and (\ref{eq:MRC_Outage_4})). Based on these equations, the performance improvement of JD over SC and MRC can be quantified in terms of the SNR gain, which is the ratio of the coding gains, 
\begin{align}
\label{eq:Gain_SC}
G_{\text{JD},\text{SC}}= & \frac{G_{\text{C},\text{JD}}}{G_{\text{C},\text{SC}}}  = \frac{A_1(R_\text{c})}{\sqrt[N]{A_N(R_\text{c})}}> \sqrt[N]{N!}, \quad \text{and} \\
\label{eq:Gain_MRC}
G_{\text{JD},\text{MRC}}= & \frac{G_{\text{C},\text{JD}}}{G_{\text{C},\text{MRC}}} = \frac{1}{\sqrt[N]{N!}} \cdot \frac{A_1(R_\text{c})}{\sqrt[N]{A_N(R_\text{c})}} > 1,
\end{align}
respectively. In  Lemma~\ref{th:A_N} (see Appendix~\ref{app:Lemma6}) we prove that the SNR gain of JD over MRC is strictly larger than one, which implies that the SNR gain of JD over SC is strictly larger than $\sqrt[N]{N!}$. 
\section{SNR gain vs. Diversity-Multiplexing Tradeoff} 
\label{DMT}
In this section we relate the SNR gain to the DMT analysis. Note that the SNR gain in (\ref{eq:SNRGain_JD}) can be separated into two parts, one depending on the spectral efficiency and the other depending on the outage probability. Both parts are influenced by the number of links. As we can easily see, with a decreasing outage probability, the SNR gain increases, scaled by the power of $(N-1)/N$, i.e.,
\begin{align}
\label{eq:SNRGain_JD_dif_Out}
\frac{  \partial 10 \log_{10} \left(G_{\text{MCo},\text{SCo}} \right)~\text{dB}}{\partial P^\text{out}} = - 4.3\frac{N-1}{N} \frac{1}{P^\text{out}}~\text{dB}.
\end{align}
The dependency of the SNR gain in terms of the spectral efficiency cannot be seen that easily. However, similarly as for the throughput (cf. (\ref{eq:Lambert_1})) we can simplify the SNR gain for sufficiently high spectral efficiencies, i.e., $R_\text{c} \gg 1$, as
\begin{align}
G_{\text{MCo},\text{SCo}} & \approx \sqrt[N]{\frac{(N-1)!}{\left(\ln(2)\right)^{N-1} N^N}}  \frac{2^{R_\text{c} \frac{N-1}{N}}}{R_\text{c}^{\frac{N-1}{N}}}  \nonumber \\
\label{eq:SNRGain_JD_approx}
& \quad \times \frac{1}{\sqrt[N]{\left(P^\text{out}\right)^{N-1}}}  \frac{\sqrt[N]{\prod_{i=1}^{N} d_i^{-\eta}}}{d_1^{-\eta}}.
\end{align}	
Now, we can see that with increasing spectral efficiency, the SNR gain increases, scaled by the factor $(N-1)/N$, i.e., 
\begin{align}
\label{eq:SNRGain_JD_dif_R_c}
\frac{  \partial 10 \log_{10} \left(G_{\text{MCo},\text{SCo}} \right)\text{dB} }{\partial R_\text{c}} \approx 3\frac{N-1}{N}~\text{dB}.
\end{align}
A similar analysis for the SNR gain of JD over SC and MRC in (\ref{eq:Gain_SC}) and (\ref{eq:Gain_MRC}), respectively, leads to the same result as in (\ref{eq:SNRGain_JD_dif_R_c}) with respect to the spectral efficiency, while being insensitive to the target outage probability.	
	
We see that in (\ref{eq:SNRGain_JD_dif_Out}) and (\ref{eq:SNRGain_JD_dif_R_c}) there is a factor of $(N-1)/N$. This factor can be related to the DMT analysis, as we show in the following. In the context of MIMO systems~\cite{Zheng2003}, it is proven that for a multiplexing gain
\begin{align}
\label{eq:multiplexing_gain}
r= & \lim\limits_{\bar{\Gamma} \rightarrow \infty}\frac{R_c(\bar{\Gamma})}{\ld\left(N \bar{\Gamma}\right)} , 
\intertext{the diversity gain $d$ will not exceed}
\label{eq:diversity_gain}
d(r)= & -\lim\limits_{\bar{\Gamma} \rightarrow \infty} \frac{\ld P^\text{out}_{\cdot,N}(r,\bar{\Gamma})}{\ld \left(N \bar{\Gamma}\right)},
\end{align}
for  $\bar{\Gamma}_1=...=\bar{\Gamma}_N=\bar{\Gamma}$, i.e., all distances are normalized to unit, and average system SNR $N \bar{\Gamma}$. By applying a singular value decomposition, the MIMO fading channel can also be transformed into a parallel fading channel in the space domain\cite{tse2005}. Thus, it is reasonable to relate the DMT with the SNR gain.\\ \indent
For JD, the diversity gain is a function of the multiplexing gain given by	
	\begin{align}
	\label{eq:diversity_gain_JD_0}
	d_\text{JD}(r)= & N-r, \qquad \qquad r \in [0,N].\\	
\intertext{The DMT for SC and MRC is given by}
	\label{eq:diversity_gain_SC_0}
	d_j(r)= & N \cdot (1-r), \qquad r \in [0,1],
	\end{align}
for $j\in\{\text{SC,\text{MRC}}\}$. For more details, we refer to the derivations in Appendix~\ref{app:Theorem4}. \\ \indent
The DMT of JD based on the lower bound in~(\ref{eq:JD_Outage_7}) given in~\cite[Ch.~9.1.3]{tse2005} is aligned with our results. It is not surprising that JD outperforms SC and MRC in terms of the multiplexing gain. Both SC and MSC perform a non-invertible linear transform on the received signal vector, collapsing the dimension from $N$ to one. It is obvious that diversity can be maintained with SC and MRC, but both will suffer with respect to JD when the goal is to achieve multiplexing gain. \\ \indent
At full diversity $(d_\text{JD}=N,r=0)$ (i.e., fixed spectral efficiency), the SNR gain increases by a factor proportional to $(N-1)/N$ with a decreasing outage probability (cf. (\ref{eq:SNRGain_JD_dif_Out})). The term $(N-1)/N$ is the relative maximum diversity gain of MCo (with JD) and SCo. 
At full multiplexing $(d_\text{JD}=0,r=N)$ (i.e., fixed outage probability), the SNR gain increases by a factor proportional to  $(N-1)/N$ with a increasing spectral efficiency (cf. (\ref{eq:SNRGain_JD_dif_R_c})). Similar to the full diversity, the term $(N-1)/N$ is the relative maximum multiplexing gain of MCo (with JD) and SCo. In summary, both performance improvements are governed by the relation of the maximum diversity and maximum multiplexing gains of MCo (with JD) and SCo.

Based on the DMT analysis, one can give the slope of the SNR gain in the spectral efficiency and outage probability, but not the SNR gain itself. Furthermore, the slope in the spectral efficiency is merely valid for sufficiently  high values, as we discuss in the next section.
\section{Numerical Results}
\label{NumericalResults}
In this section, we illustrate and discuss the exact, asymptotic, and the lower bound outage probabilities of JD as well as the exact and asymptotic throughput of JD. Furthermore, we illustrate and discuss the corresponding SNR gain. We equally allocate the total transmit power $P_\text{T}$ to all channels, such that $P_i=P_\text{T}/N, \forall i \in \mathcal{N}$. Furthermore, we normalize all distances to one. We define the average system transmit SNR as $P_\text{T}/N_0$.\\ \indent
\begin{figure*}[t]
		\centering
	\begin{subfigure}{0.45\linewidth}
		\includegraphics[width=\linewidth]{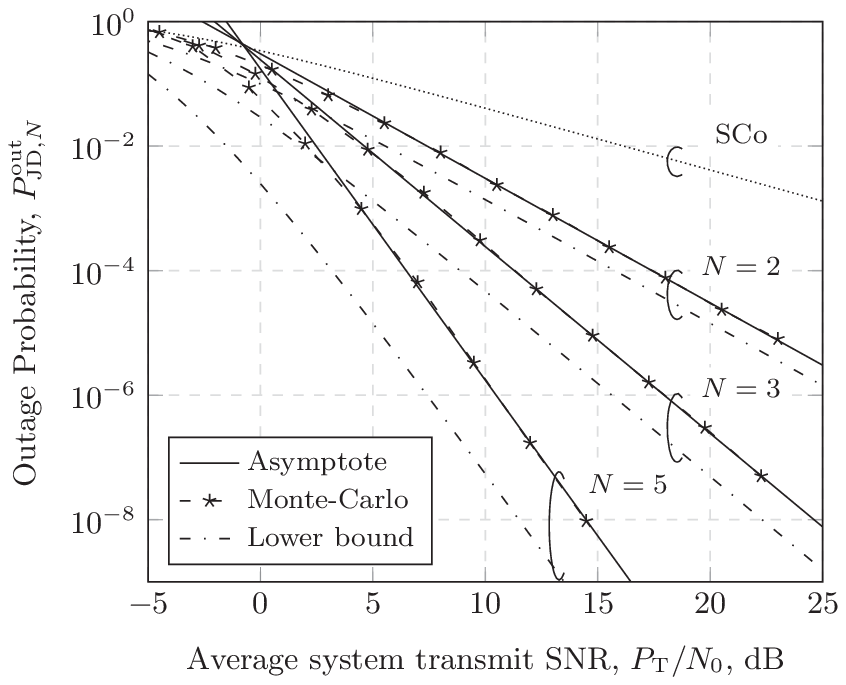}
		\caption{}
		\label{fig:Outage_MCo}
	\end{subfigure}
	\begin{subfigure}{0.45\linewidth}
		\includegraphics[width=\linewidth]{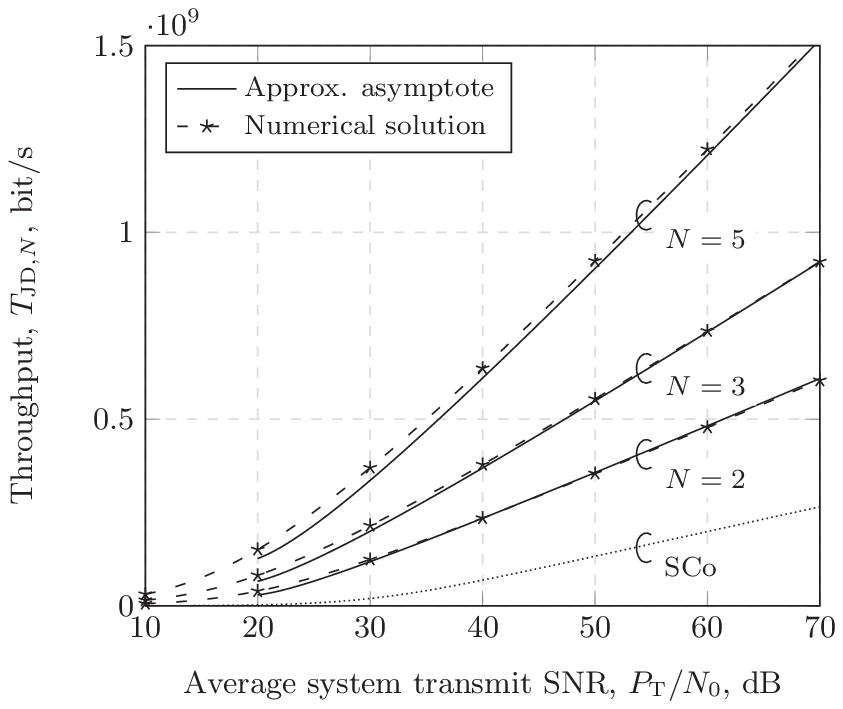}
		\caption{}
		\label{fig:Throughput_MCo}
	\end{subfigure}
	\caption{(a) JD outage probability for Monte-Carlo simulation, asymptote, and lower bound, with $N\in \{2,3,5\}$ and, $R_\text{c}=0.5$, and (b) JD throughput Monte-Carlo simulation, approximated asymptote with $N\in\{2,3,5\}$, $B=20$~MHz, and $P^\text{out}=10^{-3}$. The outage probability and throughput of SCo are depicted for comparison.}
	\label{fig:Outage_Throughput}
\end{figure*}
Fig.~\ref{fig:Outage_MCo} depicts the outage probability of JD (Monte-Carlo simulation of (\ref{eq:JD_Outage_4}), our asymptote in (\ref{eq:JD_Outage_6}), and the existing lower bound in (\ref{eq:JD_Outage_7})) versus the average system transmit SNR $P_\text{T}/N_0$. For comparison, we include the SCo outage probability in (\ref{eq:SCo_Outage_1}).  We show results for $N\in\{2,3,5\}$ and a constant spectral efficiency of $R_\text{c}=0.5$. We can observe the following: (i) the asymptote is very tight at medium and high SNR; (ii) with every additional link the diversity gain $d_\text{JD}(r)$ increases by one with constant spectral efficiency, i.e., the multiplexing gain is $r=0$; and (iii) the SNR offset of the lower bound increases with the number of links. At this point, we would like to clarify our assumptions on the SNR range. It is noteworthy that, even though our outage analysis is based on high SNR, it leads to accurate results in the low-to-medium SNR region as well. In Fig.~\ref{fig:Outage_MCo}, for instance, the asymptotic outage probability with 5 links is already tight for an outage probability of $P^\text{out}_{\text{JD},5}=10^{-3}$. The corresponding average system transmit SNR is then $P_\text{T}/N_0=5$~dB. That means that the average transmit SNR per link is $P_i/N_0=-2$~dB, which falls in the low-to-medium SNR region.\\ \indent
\emph{Remark:} The lower and upper bounds based on the outage exponent analysis in (\ref{eq:JD_Outage_8}) are tight, cf. \cite[Fig.~3 - Fig.~7]{Bai2013}, but to achieve tight bounds for the entire SNR range heavy computation efforts are required. Especially for the class of URRLC applications, the low-SNR range is not of interest, as the region of the required outage probabilities is at medium to high SNR.\\ \indent
Fig.~\ref{fig:Throughput_MCo} depicts the throughput for JD (numerical solution of (\ref{eq:system_throughput}), our asymptote in (\ref{eq:system_throughput_JD}) with the approximation of the asymptotic inverse function in (\ref{eq:Inverse_Function})) versus the average system transmit SNR $P_\text{T}/N_0$.  For comparison, we illustrate the SCo throughput in (\ref{eq:system_throughput_SCo}). We show results for $N\in\{2,3,5\}$, $B=20$~MHz, and an outage probability of $P^\text{out}=10^{-3}$. The following can be observed: (i) the approximated asymptote is very tight at high SNR; and (ii) for increasing SNR, the JD throughput increases asymptotically with $N\cdot20$~Mbit/s per $3$~dB, whereas the SCo throughput increases asymptotically with $1\cdot20$~Mbit/s per $3$~dB.\\ \indent
In the following, we show numerical results for the SNR gain with $P^\text{out}\in\{10^{-3},10^{-5}\}$ and $N\in \{2,3,4\}$. As shown in Fig.~\ref{fig:Outage_MCo}, the asymptotic outage probability is very tight within this range, i.e., the following numerical results based on the asymptotic outage probability barely differ from the numerical results based on the exact outage probability. \\ \indent
\begin{figure*}[t]
	\centering
	\begin{subfigure}{0.45\linewidth}
		\includegraphics[width=\linewidth]{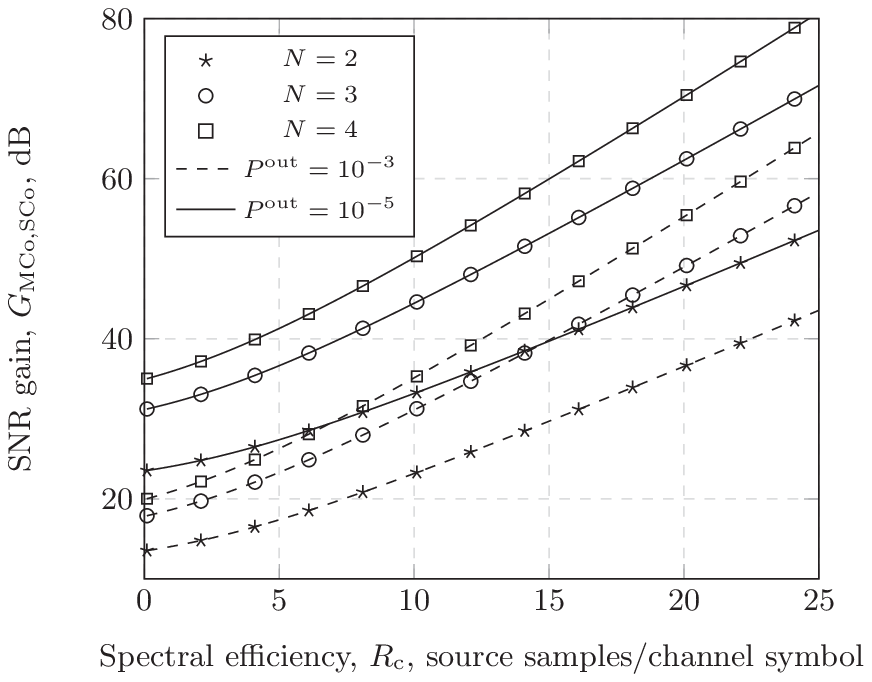}
		\caption{}
		\label{fig:Gain_SNR}
	\end{subfigure}
	\begin{subfigure}{0.45\linewidth}
		\includegraphics[width=\linewidth]{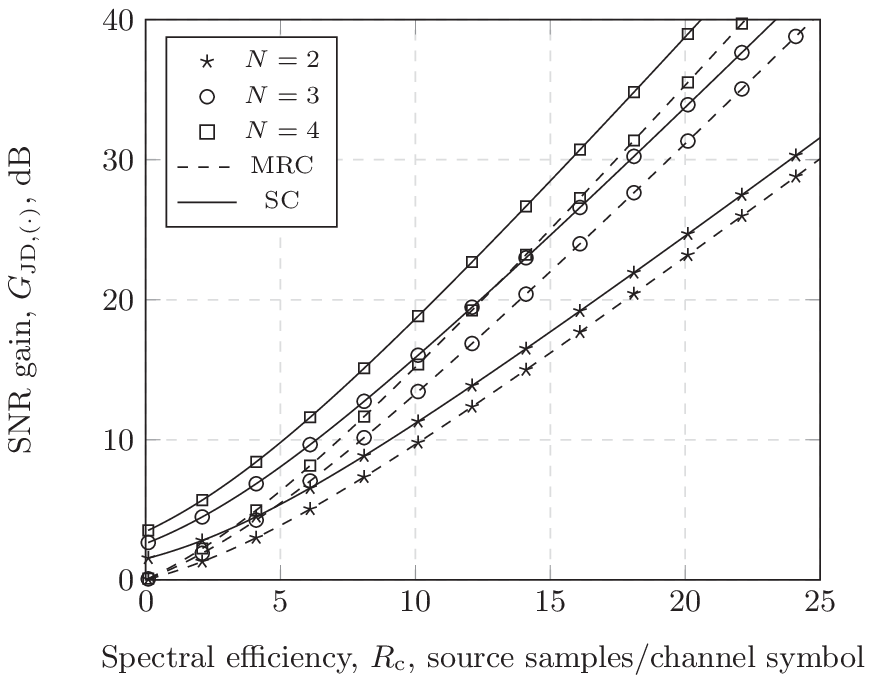}
		\caption{}
		\label{fig:Gain_Rate}
	\end{subfigure}
	\caption{(a) SNR gain of MCo (with JD) over SCo, with $N\in \{2,3,4\}$ and $P^\text{out}\in\{10^{-3},10^{-5}\}$ and (b) SNR gains of JD over SC and MRC, with $N\in \{2,3,4\}$.}
	\label{fig:Gain_MCo}
\end{figure*}
Fig.~\ref{fig:Gain_SNR} depicts the SNR gain of MCo over SCo ---$G_{\text{MCo},\text{SCo}}$ in (\ref{eq:SNRGain_JD})--- versus the spectral efficiency $R_\text{c}$ (corresponding to the throughput $T\approx B R_\text{c}$).  We show results with $N\in \{2,3,4\}$ number of links and an outage probability of $P^\text{out}\in\{10^{-3},10^{-5}\}$. The following can be observed: (i) the SNR gain increases with the number of links, spectral efficiency, and decreasing outage probability; (ii) a decrease in outage probability manifests itself as a vertical shift of the respective SNR gain, i.e., from (\ref{eq:SNRGain_JD_dif_Out}) we have $\Delta G_{\text{MCo},\text{SCo}}= 2 \cdot 10 \frac{N-1}{N}$~dB for an outage probability shift from  $10^{-3}$ to $10^{-5}$ ; and (iii) for sufficiently high spectral efficiencies, the SNR gain increases by $3(N-1)/N$~dB per source sample/channel symbol (cf. (\ref{eq:SNRGain_JD_dif_R_c})). \\ \indent
Fig.~\ref{fig:Gain_Rate} depicts the SNR gain of JD ---$G_{\text{JD},(\cdot)}$ given in (\ref{eq:Gain_SC}) and (\ref{eq:Gain_MRC})--- over SC and MRC, respectively, versus the spectral efficiency $R_\text{c}$. The following can be observed: (i) the SNR gain of JD is greater than one, as proven in Lemma~\ref{th:A_N}; (ii) the SNR gain of JD increases with $N$ and $R_c$; (iii) the SNR gain of JD with respect to MRC differs from the SNR gain of JD with respect to SC by $\sfrac{1}{\sqrt[N]{N!}}$; (iv) for very low spectral efficiencies the SNR gain of JD over MRC vanishes; and (iv) for sufficiently high spectral efficiencies, the SNR gain increase by $3(N-1)/N$~dB per source sample/channel symbol (cf. (\ref{eq:SNRGain_JD_dif_R_c})).\\ \indent
Note that the range of practical spectral efficiencies is within $R_\text{c}\in[0.5,4.\bar{6}]$ (e.g., $1/2$ channel code rate and BPSK, or $2/3$ channel code rate and $128$-QAM). However, we illustrate spectral efficiency up to $R_\text{c}=25$, in order to show the asymptotic slope of the SNR gain.

\section{Cellular Field Trial for the Uplink}
\label{FieldTrail}
In \cite{Grieger2014}, measurements were carried out in a field trial testbed in Dresden (Germany) downtown. In this section we make use of this measurement data to show the potential of MCo in a real cellular network. In the field trial, the uplink was considered. Next we shortly introduce the field trial setup and then present the empirical outage probability and throughput cumulative distribution functions (CDFs) for the measurement data. Our results elaborate on the following points: (i) the performance improvement of MCo over SCo and (ii) the performance gain of JD in comparison to SC and MRC.  
\begin{figure*}[t]
	\centering
		\begin{subfigure}{0.35\linewidth}
	\includegraphics[width=\linewidth,trim={2.5cm 6.5cm 2.5cm 7.5cm},clip]{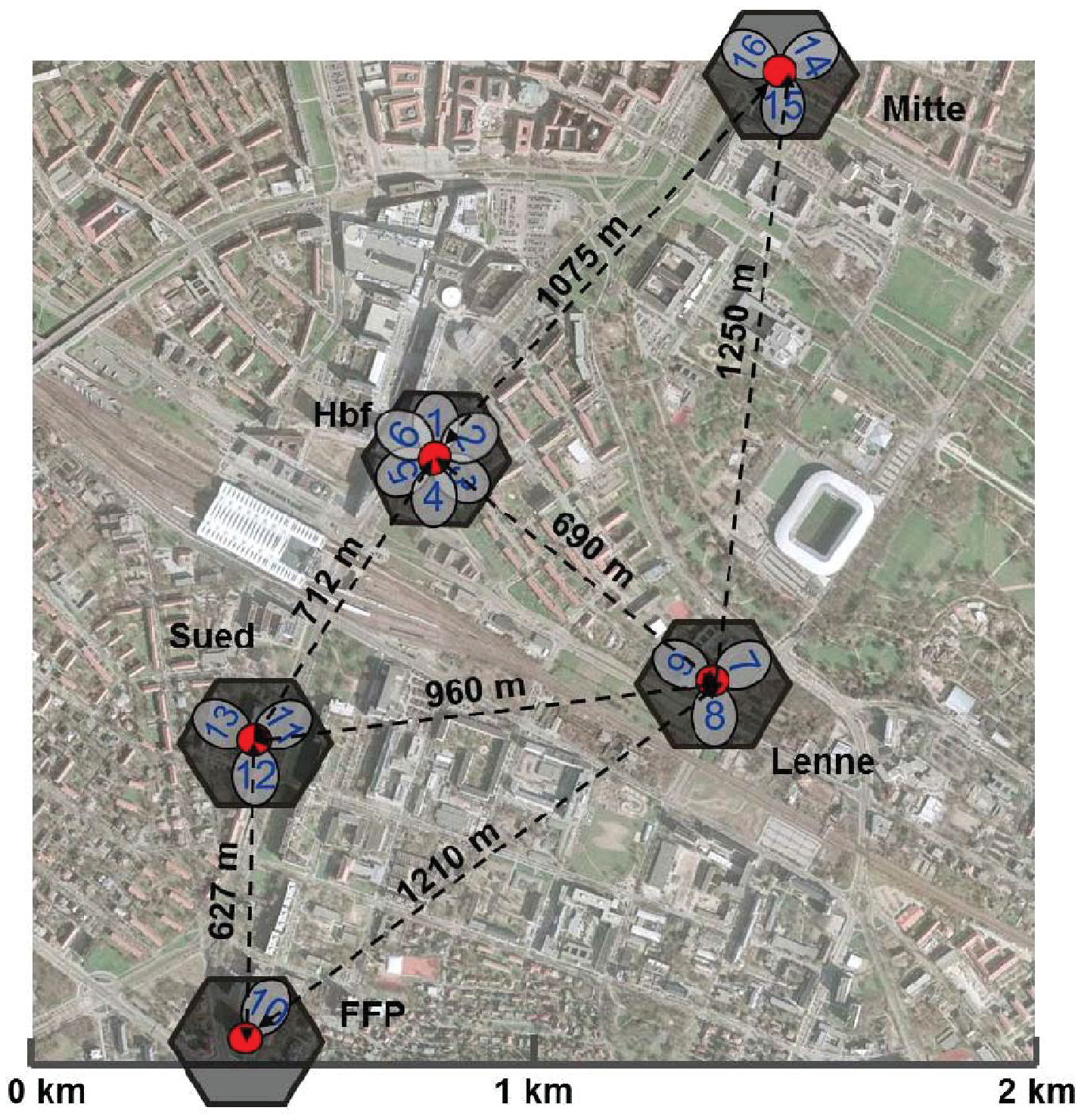}
			\caption{}
			\label{fig:Testbed}
		\end{subfigure}
		\begin{subfigure}{0.52\linewidth}
			\includegraphics[width=\linewidth]{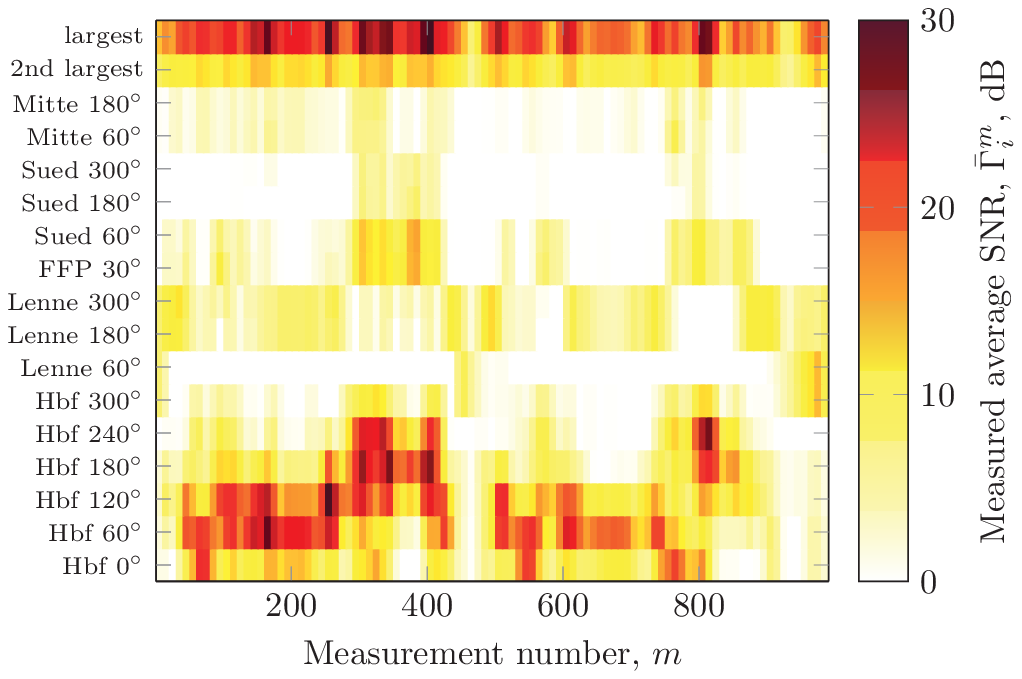}
			\caption{}
			\label{fig:SNRData}
		\end{subfigure}	
		\begin{subfigure}{0.45\linewidth}
			\includegraphics[width=\linewidth]{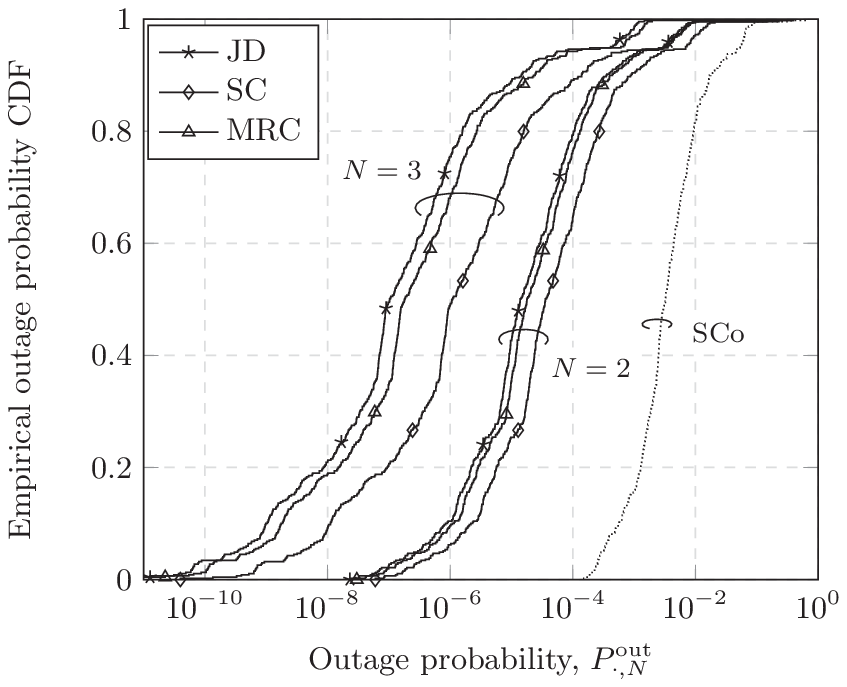}
			\caption{}
			\label{fig:CDF_Outage}
		\end{subfigure}		
			\begin{subfigure}{0.45\linewidth}
		\includegraphics[width=\linewidth]{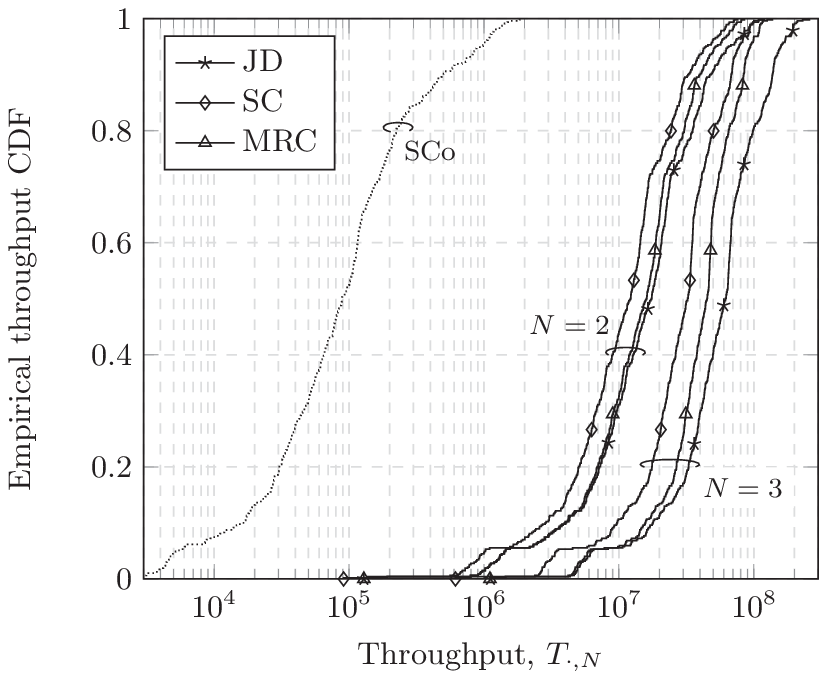}
		\caption{}
		\label{fig:CDF_Throughput}
	\end{subfigure}	
	\caption{(a) Testbed deployment, (b) measured average SNR $\bar{\Gamma}_i^m$ achieved at all BSs of the testbed during the complete field trial \cite{Grieger2014}, (c) empirical outage probability CDF ($N\in \{2,3\}$, $R_\text{c}=1$), and (d)  empirical throughput CDF ($N\in \{2,3\}$, $P^\text{out}=10^{-5}$) for JD, SC, and MRC.}
\end{figure*}
\subsection{Field Trial Setup}
The field trial testbed, deployed in Dresden downtown, is depicted in Fig.~\ref{fig:Testbed}. In total, 16 BSs located on five sites with up to six-fold sectorization were used for the measurements. During the field trial, two UEs were moved on a measurement bus in 5 m distance while transmitting on the same time and frequency resources employing one dipole antenna each. The superimposed signal is jointly received by all BSs, which took snapshots of 80 ms (corresponding to 80 transmit time intervals) every 10 s. In total, about 1900 such measurements were taken in order to observe a large number of different transmission scenarios. In Fig. \ref{fig:SNRData} the measured average SNR $\bar{\Gamma}_i^m$ values for around 1000 measurements observed at all BSs and locations are shown, where $m$ denotes the measurement number and $i$ the BS index, $i \in \{\text{Hbf } 0\si{\degree},\text{Hbf } 60\si{\degree},... \}$. The two largest average SNRs measured at any BS for each measurement are depicted in the upper part of the figure. An interesting result is that multiple relatively high average SNR values of two different BSs  are observed at each location of the UEs. Since combining algorithms are particularly beneficial in scenarios with multiple relatively high average SNR values, this result indicates that cooperation among BSs can provide a much more reliable data transmission, as confirmed next. For more details on this field trial setup, please refer to \cite{Grieger2014}.
\subsection{Empirical CDFs for Outage Probability and Throughput}
With the measured average SNR $\bar{\Gamma}_i^m$ in Fig. \ref{fig:SNRData} we can generate an empirical outage probability and throughput CDF. For each measurement we consider the $N$ strongest links, i.e., the largest measured average SNRs $\bar{\Gamma}_i^m$. The outage probability can be assessed with (\ref{eq:JD_Outage_6}), (\ref{eq:SC_Outage_3}) and (\ref{eq:MRC_Outage_9}) for JD, SC, and MRC, respectively, for each measurement. Similarly, the throughput is given by (\ref{eq:system_throughput_JD}), (\ref{eq:system_throughput_SC}), and (\ref{eq:system_throughput_MRC}) for JD, SC, and MRC, respectively.\\ \indent  
Fig.~\ref{fig:CDF_Outage} and Fig.~\ref{fig:CDF_Throughput} depict the empirical outage probability CDF and the empirical throughput CDF, respectively. We show results for $N\in\{2,3\}$ number of links, a spectral efficiency of $R_\text{c}=1$ (Fig.~\ref{fig:CDF_Outage}), and an outage probability of $P^\text{out}=10^{-5}$ (Fig.~\ref{fig:CDF_Throughput}). The following can be observed: (i) MCo is much superior to SCo and (ii) JD outperforms SC and MRC, the performance gain increasing with the number of links from source to destination.
\subsection{Discussion}
Based on the field trial setup, we can conclude that MCo can achieve a substantial performance improvement in real cellular networks. The measurement data at hand documents that multiple relatively high average SNR values frequently occur, for which combining algorithms are particularly beneficial. From the uplink measurement data we can also draw conclusions for the downlink. As argued in Section~\ref{LinkModel}, the statistical properties of the link model are identical for the up- and downlinks. Thus, it is reasonable to assume that multiple relatively high average SNR values frequently occur in the downlink as well. Based on this, MCo can also achieve low outage probabilities and high throughput in the downlink of real cellular networks.
\section{Conclusion}
In this work, we answer fundamental questions regarding the performance improvement of multi-connectivity over single-connectivity. For doing so, we derive the exact and asymptotic outage probabilities for optimal and suboptimal receiver algorithms, namely joint decoding, selection combining, and maximal-ratio combining. We evaluate and show the tremendous transmit power reduction of MCo over SCo, by analytically deriving the corresponding SNR gain. Our results reveal that the SNR gain increases with the number of links and the spectral efficiency, while decreasing with the outage probability. \\ \indent
In addition, we compare the optimal combining algorithm, i.e, joint decoding, with the suboptimal combining algorithms, i.e., selection combining and maximal-ratio combining. Again, we quantify the performance gain in terms of the SNR reduction. Our results reveal that JD becomes more advantageous in terms of the SNR gain as the spectral efficiency and the number of links increases.\\ \indent
Finally, we have applied our analytical framework to real cellular networks. Based on the measurement data recorded in a field trial, we have evaluated the achievable performance improvement by the use of multi-connectivity. The measurement data documents that multiple relatively high average SNR values frequently occur, in which case multi-connectivity proves particularly beneficial.
\section*{Acknowledgments}
The authors wish to acknowledge the helpful discussions with David {\"O}hmann and Nick Schwarzenberg.
\appendices
\section{Asymptotic Joint Decoding Outage Probability}
\label{app:ExactOutage}
The asymptotic outage probability can be obtained as
{\allowdisplaybreaks\begin{align}
& P_{\text{JD},N}^{\text{out}} = \int_{\gamma_1=0}^{2^{R_\text{c} }-1} \int_{\gamma_2=0}^{2^{ R_\text{c}  - \phi(\gamma_1)}-1} ... \nonumber \\
& \quad \int_{\gamma_N=0}^{2^{ R_\text{c}  - \phi(\gamma_1) - \phi(\gamma_2) - ... -  \phi(\gamma_{N-1})}-1}  \frac{1}{\bar{\Gamma}_1} \exp\Big(-\frac{\gamma_1}{\bar{\Gamma}_1}\Big) \nonumber \\ 
\label{eq:JD_Outage_app2}
& \quad \times \frac{1}{\bar{\Gamma}_2} \exp\Big(-\frac{\gamma_2}{\bar{\Gamma}_2}\Big) ... \frac{1}{\bar{\Gamma}_N} \exp\Big(-\frac{\gamma_N}{\bar{\Gamma}_N}\Big) \text{d} \gamma_N ... \text{d} \gamma_2 \text{d} \gamma_1 \\
& \approx \int_{\gamma_1=0}^{2^{R_\text{c} }-1} \int_{\gamma_2=0}^{2^{ R_\text{c}  - \phi(\gamma_1)}-1} ... \int_{\gamma_N=0}^{2^{ R_\text{c}  - \phi(\gamma_1) - \phi(\gamma_2) - ... -  \phi(\gamma_{N-1})}-1} \nonumber \\ 
\label{eq:JD_Outage_app3}
& \quad \frac{1}{\bar{\Gamma}_1} \Big(1- \frac{\gamma_1}{\bar{\Gamma}_1}\Big) \frac{1}{\bar{\Gamma}_2} \Big(1- \frac{\gamma_2}{\bar{\Gamma}_2}\Big) ... \frac{1}{\bar{\Gamma}_N} \Big(1- \frac{\gamma_N}{\bar{\Gamma}_N}\Big) \text{d} \gamma_N ... \text{d} \gamma_2 \text{d} \gamma_1 \\
& \approx  \frac{1}{\bar{\Gamma}_1 \bar{\Gamma}_2 ... \bar{\Gamma}_N} \int_{\gamma_1=0}^{2^{R_\text{c} }-1} \int_{\gamma_2=0}^{2^{ R_\text{c}  - \phi(\gamma_1)}-1} ... \nonumber \\
\label{eq:JD_Outage_app4}
& \quad \int_{\gamma_N=0}^{2^{ R_\text{c}  - \phi(\gamma_1) - \phi(\gamma_2) - ... -  \phi(\gamma_{N-1})}-1}  \text{d} \gamma_N ... \text{d} \gamma_2 \text{d} \gamma_1 \\
\label{eq:JD_Outage_app5}
&= \frac{A_N(R_\text{c})}{\bar{\Gamma}_1 \bar{\Gamma}_2 ... \bar{\Gamma}_N}.
\end{align}}\noindent
The steps are justified as follows: (\ref{eq:JD_Outage_app2}) is the substitution of the pdf $f(\gamma_i)$ given in (\ref{eq:RFpdf}) into (\ref{eq:JD_Outage_4}); (\ref{eq:JD_Outage_app3}) MacLaurin series for exponential function  $\exp(-x_i) \approx 1- x_i$ for $x_i\rightarrow 0$, (\ref{eq:JD_Outage_app4}) expanding the resulting product as $\prod\nolimits_{i}(1-x_i) \approx 1$ for $x_i\rightarrow 0$; (\ref{eq:JD_Outage_app5}) is proven in Lemma~\ref{th:highSNR} and the assumption that the received SNRs are independently distributed, thus we can interchange the integral bounds. 
\begin{lemma}\label{th:highSNR}
 For any $N \in \mathbb{N}\backslash \{1\}$,
	\begin{align}
	A_N(x) & = \int_{\gamma_N=0}^{2^{x}-1}  \int_{\gamma_{N-1}=0}^{2^{x-\phi(\gamma_N)}-1} ...  \int_{\gamma_{1}=0}^{2^{x-\phi(\gamma_N)-...-\phi(\gamma_2)}-1} \nonumber \\
\label{eq:theorem_3_1}
	&  \quad \text{d} \gamma_1 ... \text{d} \gamma_{N-1} \text{d} \gamma_N \\
\label{eq:theorem_3_2}
		 & =  (-1)^N\left(1 - 2^x \cdot e_N  \left(-x \ln(2)\right) \right).
\end{align}

Here, $e_N(y)=\sum_{n=0}^{N-1}\frac{y^n}{n!}$ is the exponential sum function.
\end{lemma}
\begin{proof}
\emph{Base case:} If $N=2$, then (\ref{eq:theorem_3_1}) is
{\allowdisplaybreaks	\begin{align}
	A_2(x) & =  \int_{\gamma_2=0}^{2^{x}-1}  \int_{\gamma_{1}=0}^{2^{x-\phi(\gamma_2)}-1} 1 \text{d} \gamma_1 \text{d} \gamma_2 \nonumber \\
	& = \int_{\gamma_2=0}^{2^{x}-1} \left[\frac{2^{x}}{1+\gamma_2}-1\right] \text{d} \gamma_2 \nonumber \\
	& = 2^x \left(x \cdot \ln(2)-1\right)+1,
\end{align}	}\noindent
which is (\ref{eq:theorem_3_2}) for $N=2$. So, the theorem holds for $N=2$. \\ \indent
\emph{Inductive hypothesis:} Suppose the theorem holds for all values of $N$ up to some $K \geq 2$. \\ \indent
\emph{Inductive step:} Let $N=K+1$, then (\ref{eq:theorem_3_1}) is
{\allowdisplaybreaks\begin{align}
& A_{K+1}(x) =  \int_{\gamma_{K+1}=0}^{2^{x}-1} \nonumber \\
& \quad \underbrace{\int_{\gamma_{K}=0}^{2^{x-\phi(\gamma_{K+1})}-1} ...  \int_{\gamma_{1}=0}^{2^{x-\phi(\gamma_{K+1})-...-\phi(\gamma_2)}-1} 1 \text{d} \gamma_1 ... \text{d} \gamma_{K}}_{A_K(x-\phi(\gamma_{K+1}))} \nonumber \\
\label{eq:proof_3_1}
& \quad \text{d} \gamma_{K+1} \\
& =  \int_{\gamma_{K+1}=0}^{2^{x}-1} \bigg[(-1)^K+\frac{2^{x}}{1+\gamma_{K+1}} \sum\nolimits_{n=0}^{K-1} (-1)^{K+n+1} \nonumber \\
\label{eq:proof_3_2}
& \quad \times \frac{1}{n!} \left(x-\phi(\gamma_{K+1})\right)^n \left( \ln(2) \right)^n\bigg] \text{d} \gamma_{K+1} \\
& =  \int_{\gamma_{K+1}=0}^{2^{x}-1} \bigg[(-1)^K+\frac{2^{x}}{1+\gamma_{K+1}} \sum\nolimits_{n=0}^{K-1} (-1)^{K+n+1} \frac{1}{n!}\nonumber \\
\label{eq:proof_3_3}
& \quad \times  \left( \ln(2) \right)^n   \sum\nolimits_{k=0}^{n} (-1)^k \binom{n}{k}  x^{n-k} \phi(\gamma_{K+1})^k \bigg] \text{d} \gamma_{K+1} \\
& =  (-1)^K \gamma_{K+1} +2^{x}\sum\nolimits_{n=0}^{K-1} (-1)^{K+n+1} \frac{1}{n!} \left( \ln(2) \right)^n  \nonumber \\
\label{eq:proof_3_4}
& \quad \times \sum\nolimits_{k=0}^{n} (-1)^k \binom{n}{k} x^{n-k} \frac{\left(\ln(1+\gamma_{K+1})\right)^{k+1}}{(k+1) \left(\ln(2)\right)^k} \bigg|_{\gamma_{K+1}=0}^{2^x-1} \\
& =  (-1)^K\bigg( 2^x -1 -2^{x}\sum\nolimits_{n=0}^{K-1} (-1)^{n} \frac{1}{n!} x^{n+1} \left(\ln(2) \right)^{n+1} \nonumber \\
\label{eq:proof_3_5}
& \quad \times \underbrace{\sum\nolimits_{k=0}^{n} (-1)^k \binom{n}{k} \frac{1}{k+1}}_{\sfrac{1}{(n+1)}} \bigg) \\
\label{eq:proof_3_6}
& = (-1)^{K+1}\left(1 - 2^{x} \cdot e_{K+1}  \left(-x \ln(2)\right) \right).
\end{align}}\noindent
The steps are justified as follows: (\ref{eq:proof_3_2}) is our inductive hypothesis; for (\ref{eq:proof_3_3}) we have used the binomial formula; (\ref{eq:proof_3_6}) we have used the following
\begin{align}
&	\sum\nolimits_{k=0}^{n} (-1)^k \binom{n}{k} \frac{1}{k+1} = \frac{1}{n+1} \sum\nolimits_{k=0}^{n} (-1)^k \binom{n+1}{k+1} \\
&	= \frac{-1}{n+1} \bigg[\underbrace{\sum\nolimits_{k=1}^{n+1} (-1)^k \binom{n+1}{k} + \binom{n+1}{0}}_{\sum\nolimits_{k=0}^{n+1} (-1)^k \binom{n+1}{k}=(1-1)^{n+1}=0} - \binom{n+1}{0}\bigg] \\
& =\frac{1}{n+1}
\end{align}
and carried out some algebraic manipulations. Equation (\ref{eq:proof_3_6}) corresponds to (\ref{eq:theorem_3_2}) for $N=K+1$. So, the theorem holds for $N=K+1$. Hence, by the principle of mathematical induction, the theorem holds for all $N \in \mathbb{N} \backslash \{1\}$. 
\end{proof}
\section{Lemma~\ref{th:UpperBound}}
\label{app:UpperBound}
\begin{lemma}\label{th:UpperBound}
	For any $N\in \mathbb{N}$, $R_\text{c}>0$, and RVs $\Gamma_{i}\in \mathds{R}_+, \forall i \in \mathcal{N}=\{1,...,N\}$,
	\begin{align}
	\label{eq:UpperBound}
	\emph{Pr}\left[0\leq \sum\nolimits_{i=1}^{N} \Gamma_i \leq A_1(R_\text{c})\right] \leq  \frac{1}{N!}\frac{\left(A_1(R_\text{c})\right)^N}{\prod_{i=1}^{N}\bar{\Gamma}_i} ,
	\end{align}
	where $A_1(R_\text{c})=2^{R_\text{c}}-1$ and pdf $f_{\Gamma_{i}}(\gamma_i)=\sfrac{1}{\bar{\Gamma}_i}\exp\left(-\sfrac{\gamma_i}{\bar{\Gamma}_{i}}\right),$ $\forall i \in \mathcal{N}$.
\end{lemma}
\begin{proof} 
	Let us define a $N$-fold simplex as
	\begin{align}
		\mathcal{S}_N & := \bigg\{(x_1,...,x_N)\in \mathds{R}^N: x_i\geq 0, \nonumber \\
\label{eq:lemma2_1}	
		& \qquad \sum\nolimits_{i=1}^{N}  \bar{\Gamma}_i x_i \leq A_1(R_\text{c}) \bigg\}.
	\end{align}
	The geometric simplex volume is \cite{ellis1976}
	\begin{align}
\label{eq:lemma2_2}	
		\text{Vol}\left(	\mathcal{S}_N \right):=  \idotsint\limits_{\mathcal{S}_N}	\text{d} x_N ... \text{d} x_1= \frac{1}{N!}\frac{\left(A_1(R_\text{c})\right)^N}{\prod_{i=1}^{N}\bar{\Gamma}_i}.
	\end{align}
Now, the left hand side of Lemma~\ref{th:UpperBound} can be rewritten as
	\begin{align}
	&	\text{Pr}\left[0\leq \sum\nolimits_{i=1}^{N} \Gamma_i \leq A_1(R_\text{c})\right] \nonumber \\
\label{eq:lemma2_3}
	& = 		\text{Pr}\left[0\leq \sum\nolimits_{i=1}^{N} \bar{\Gamma}_i X_i \leq A_1(R_\text{c})\right] \\
\label{eq:lemma2_4}
		& = \idotsint\limits_{\mathcal{S}_N}	\prod\nolimits_{i=1}^{N}\exp(-x_i) 	\text{d} x_N ... \text{d} x_1	\\
\label{eq:lemma2_5}	
	& \leq \idotsint\limits_{\mathcal{S}_N}	\text{d} x_N ... \text{d} x_1.
\end{align}
The steps can be justified as follows: in (\ref{eq:lemma2_3}) and (\ref{eq:lemma2_4}) we introduce the RV $X_i=\sfrac{\Gamma_{i}}{\bar{\Gamma}_i}$, with pdf $f_{X_{i}}(x_i)=\exp\left(-x_i\right), \forall i \in \mathcal{N}$; in (\ref{eq:lemma2_5}) the exponential function can be upper bounded by $\exp(-x)\leq 1$ for $x\geq 0$; (\ref{eq:lemma2_5}) is the geometric simplex volume in (\ref{eq:lemma2_2}).	
\end{proof}
\section{Inverse Function}
\label{app:InverseFunction}
For $R_\text{c}\gg 1$, the term of the exponential sum function with the highest order in (\ref{eq:Constant_JD}) is dominant. Thus, we can approximate $A_N(R_\text{c}=x)\approx y=f(x)$ to
\begin{align}
\label{eq:Lambert_1}
y = 2^{x} x^{(N-1)}\frac{(\ln(2))^{N-1}}{(N-1)!}=2^x x^a b.
\end{align}
We can reformulate (\ref{eq:Lambert_1}) as follows
\begin{align}
	& 2^x x^a  =	\frac{y}{b}, \quad \Leftrightarrow \quad \frac{x}{a}	2^{\frac{x}{a}} = \frac{1}{a}\sqrt[a]{\frac{y}{b}}, \nonumber \\ 
	\Leftrightarrow \quad & \frac{x\ln(2)}{a} \exp\left(\frac{x\ln(2)}{a}\right)   = \frac{\ln(2)}{a}\sqrt[a]{\frac{y}{b}},\\	
	 \Leftrightarrow \quad & z\exp(z) = \frac{\ln(2)}{a}\sqrt[a]{\frac{y}{b}}, \nonumber \\ 
	 \label{eq:Lambert_6}
	 \Leftrightarrow	\quad & x  =\frac{az}{\ln(2)}=\frac{a}{\ln(2)} W\left(\frac{\ln(2) }{a}\sqrt[a]{\frac{y}{b}}\right),
\end{align}
where $W(\cdot)$ is the Lambert $W$ function, i.e., $z=g^{-1}(z\exp(z))=W(z\exp(z))$ is the inverse function of $g(z)=z\exp(z)$, for $z\geq -1$. For $z\geq e$, the Lambert $W$ function is bounded by \cite[Theorem 2.7]{hoorfar2007}
\begin{align}
\label{eq:Lambert_7}
	W(z)=\ln(z)-\ln(\ln(z)) +\Theta\left(\frac{\ln(\ln(z))}{\ln(z)}\right).
\end{align}
Finally, with (\ref{eq:Lambert_6}) and (\ref{eq:Lambert_7}) we can give an approximation of the inverse function as
\begin{align}
	A_N^{-1}(A_N) & =R_\text{c} \approx \frac{N-1}{\ln(2)} \left[\ln(\zeta)-\ln(\ln(\zeta))\right], \\
	\intertext{where }
	\zeta &=\frac{\sqrt[N-1]{ (N-1)! A_N}}{N-1}.
\end{align}
\section{Lemma~\ref{th:A_N}}
\label{app:Lemma6}
\begin{lemma}\label{th:A_N}
 For any $N\in \mathbb{N}\backslash\{1\}$ and $x>0$,
 \begin{align}
 \label{eq:lemma6_1}
 	X_N(x)=\left(A_1(x)\right)^N > N! \cdot  A_N(x)=Y_N(x),
 \end{align}
 where $A_N(x)$ is given in Lemma~\ref{th:highSNR}.
\end{lemma}
\begin{proof} 
For $x=0$ we have $X_N(0)=0$ and $Y_N(0)=0$ in (\ref{eq:lemma6_1}). Next, show that the slope of $X_N(x)$ is larger than the slope of $Y_N(x)$ for $x>0$ and thus $X_N(x)>Y_N(x)$, $x\geq0, \forall N$. 
{\allowdisplaybreaks\begin{align}
	& \frac{\text{d}}{\text{d}x} X_N(x)  = \frac{\text{d}}{\text{d}x} \left[\left(2^{x}-1\right)^N\right] \nonumber \\
\label{eq:lemma6_2}
	 & = N 2^{x} \ln (2)\left(2^{x}-1\right)^{N-1},  \\
\label{eq:lemma6_3}	
& \frac{\text{d}}{\text{d}x} Y_N(x) = \frac{\text{d}}{\text{d}x} \Big[ N! (-1)^N\Big(1-2^{x} \sum\nolimits_{n=0}^{N-1}\frac{1}{n!} \left(-x\ln(2)\right)^n\Big)\Big] \\
	& = N! \ln(2) 2^{x}(-1)^{N+1}\Big(\sum\nolimits_{n=0}^{N-1}(-1)^{n} \frac{1}{n!} \left( x\ln(2) \right)^n \nonumber \\ 
\label{eq:lemma6_4}
	& \quad + \underbrace{\sum\nolimits_{n=1}^{N-1}(-1)^n \frac{1}{(n-1)!} \left( x\ln(2) \right)^{n-1}}_{-\sum\nolimits_{n=0}^{N-2}(-1)^n \frac{1}{n!} \left( x\ln(2) \right)^{n}} \Big) \\
\label{eq:lemma6_5}
	& = N! \ln(2) 2^{x}(-1)^{2N} \frac{1}{(N-1)!} \left( x \ln (2)\right)^{N-1} \\
	& = N  2^{x} \ln(2) \left( x \ln (2)\right)^{N-1} .
	\end{align}}\noindent
We have to show that
\begin{align}
\label{eq:lemma6_6}
	\left(2^{x}-1\right)^{N-1}  & > \left( x \ln (2)\right)^{N-1} \quad \text{for } x>0.
\intertext{Since both sides in (\ref{eq:lemma6_6}) are equal for $x=0$ and have the same exponent, it is sufficient to show }
\label{eq:lemma6_7}
	\frac{\text{d}}{\text{d}x} \left(2^{x}-1\right) & > \frac{\text{d}}{\text{d}x} \left( x \ln (2)\right) \\
\label{eq:lemma6_8}
 \ln(2) 2^{x} & > \ln(2).
\end{align}
(\ref{eq:lemma6_8}) holds for $x > 0$. \end{proof}
\section{Derivation of Diversity-Multiplexing Tradeoff}
\label{app:Theorem4}
Based on the outage probability analysis considered before the diversity gains for JD, SC, and MRC are given by
	{\allowdisplaybreaks\begin{align}
		\label{eq:diversity_gain_JD_1}
		& d_\text{JD}(r)= - \lim\limits_{\bar{\Gamma} \rightarrow \infty} \frac{\ld \left(A_N(r,\bar{\Gamma})\right)- N \ld\left(\bar{\Gamma}\right)}{\ld \left(N \bar{\Gamma}\right)}\\
		& = - \lim\limits_{\bar{\Gamma} \rightarrow \infty} \bigg(r \frac{\ld\left(N \bar{\Gamma}\right)}{\ld\left(N \bar{\Gamma}\right)} \nonumber \\
		& \quad + \frac{\ld\big(\frac{1}{(N-1)!} \left(r\ld(N\bar{\Gamma})\right)^{(N-1)} (\ln(2))^{(N-1)}\big)}{\ld\left(N \bar{\Gamma}\right)} \nonumber \\
		\label{eq:diversity_gain_JD_2}
		& \quad - \frac{N \ld\left(\bar{\Gamma}\right)}{\ld\left(N \bar{\Gamma}\right)} \bigg) \\
		\label{eq:diversity_gain_JD_3}
		& = N-r,\\
		\label{eq:diversity_gain_SC_1}
		& d_\text{SC}(r)=  - \lim\limits_{\bar{\Gamma} \rightarrow \infty} \frac{N\ld \left(A_1(r,\bar{\Gamma})\right)- N \ld\left(\bar{\Gamma}\right)}{\ld \left(N \bar{\Gamma}\right)} \\
		\label{eq:diversity_gain_SC_2}		
		& =  -N  \lim\limits_{\bar{\Gamma}\rightarrow \infty} \bigg(r \frac{\ld\left(\bar{\Gamma}\right)}{\ld\left(N \bar{\Gamma}\right)}- \frac{\ld\left(\bar{\Gamma}\right)}{\ld\left(N \bar{\Gamma}\right)}\bigg) \\
		\label{eq:diversity_gain_SC_3}
		& = N \cdot (1-r), \quad \text{and} \\
		\label{eq:diversity_gain_MRC_1}
		& d_\text{MRC}(r)=  - \lim\limits_{\bar{\Gamma} \rightarrow \infty} \frac{N\ld \left(A_1(r,\bar{\Gamma})\right)- \ld\left(N!\right) -N \ld\left(\bar{\Gamma}\right)}{\ld \left(N \bar{\Gamma}\right)} \\
		\label{eq:diversity_gain_MRC_2}	
		& = -N  \lim\limits_{\bar{\Gamma}\rightarrow \infty} \bigg(r  \frac{\ld\left(\bar{\Gamma}\right)}{\ld\left(N \bar{\Gamma}\right)}- \frac{\ld\left(N!\right)}{N\ld\left(N \bar{\Gamma}\right)}   - \frac{\ld\left(\bar{\Gamma}\right)}{\ld\left(N \bar{\Gamma}\right)} \bigg) \\
		\label{eq:diversity_gain_MRC_3}
		& =  N \cdot (1-r),
		\end{align}}\noindent
respectively. The steps can be justified as follows: (\ref{eq:diversity_gain_JD_1}), (\ref{eq:diversity_gain_SC_1}), and (\ref{eq:diversity_gain_MRC_1}) are given by substituting (\ref{eq:JD_Outage_6}), (\ref{eq:SC_Outage_4}), (\ref{eq:MRC_Outage_4}) into (\ref{eq:diversity_gain}); in (\ref{eq:diversity_gain_JD_2}) we use the infinite SNR properties of $A_N\left(r,\bar{\Gamma}\right)$ in (\ref{eq:A_N}) and some algebraic manipulations; for infinite SNR the properties (\ref{eq:P_1}) and (\ref{eq:P_2}) hold which yields (\ref{eq:diversity_gain_JD_3});  in (\ref{eq:diversity_gain_SC_2}) and (\ref{eq:diversity_gain_MRC_2}) we use the infinite SNR property of $A_1\left(r,\bar{\Gamma}\right)$ in (\ref{eq:A_1}) and some algebraic manipulations; for infinite SNR the properties in (\ref{eq:P_1}) and (\ref{eq:P_2}) hold which yields (\ref{eq:diversity_gain_SC_3}) and (\ref{eq:diversity_gain_MRC_3}). \\ \indent
Substituting (\ref{eq:multiplexing_gain}) into (\ref{eq:Constant_JD}) the constants $A_N(r,\bar{\Gamma})$ and its special case $A_1(r,\bar{\Gamma})$ can be given depending on the multiplexing gain $r$ by 
	{\allowdisplaybreaks\begin{align}
		\label{eq:A_1}
		& \lim\limits_{\bar{\Gamma} \rightarrow \infty}A_1(r,\bar{\Gamma})=   \lim\limits_{\bar{\Gamma} \rightarrow \infty} \big( 2^{r \ld\left(\bar{\Gamma}\right)}-1 \big)= \lim\limits_{\bar{\Gamma} \rightarrow \infty} 2^{r \ld\left(\bar{\Gamma}\right)},\\
		& \lim\limits_{\bar{\Gamma} \rightarrow \infty} A_N(r,\bar{\Gamma})=  \lim\limits_{\bar{\Gamma} \rightarrow \infty} \Big((-1)^N+2^{r \ld\left(N\bar{\Gamma}\right)} \nonumber \\
		& \quad \sum\nolimits_{n=0}^{N-1} (-1)^{N+n+1} \frac{1}{n!} \left(r\ld(N\bar{\Gamma})\right)^n (\ln(2))^n \Big)\\
		\label{eq:A_N}
		& = \lim\limits_{\bar{\Gamma} \rightarrow \infty}   2^{r \ld\left(N\bar{\Gamma}\right)} \frac{1}{(N-1)!} \left(r\ld(N\bar{\Gamma})\right)^{(N-1)} (\ln(2))^{(N-1)},
		\end{align}}\noindent
where (\ref{eq:A_N}) can be justified with the infinite SNR properties in (\ref{eq:P_4}). Further properties for infinite SNR are:
	{\allowdisplaybreaks\begin{align}
		\label{eq:P_1}
		& \lim\limits_{\bar{\Gamma} \rightarrow \infty}  \frac{ \ld\left(\bar{\Gamma}\right)}{\ld\left(N \bar{\Gamma}\right)} =  1, \quad \lim\limits_{\bar{\Gamma} \rightarrow \infty}  \frac{\ld(N!)}{N\ld\left(N \bar{\Gamma}\right)} =  0	, \\
		\label{eq:P_2}
		& \lim\limits_{\bar{\Gamma} \rightarrow \infty}  \frac{(N-1)\ld\left( \ld\left(N\bar{\Gamma}\right)\right)}{\ld\left(N \bar{\Gamma}\right)} = 0, \\
		\label{eq:P_4}	
		& \lim\limits_{\bar{\Gamma} \rightarrow \infty}   \left(\ld(N \bar{\Gamma})\right)^{(N-1)} \gg  \lim\limits_{\bar{\Gamma} \rightarrow \infty}  \left( \ld(N \bar{\Gamma})\right)^{(N-n)}
		\end{align}}\noindent
	 for $n=2,...,N$.
\bibliographystyle{IEEEtran}
\bibliography{refs_WSN}

\end{document}